\documentclass{LMCS}

\def\dOi{10(2:3)2014}
\lmcsheading%
{\dOi}
{1--35}
{}
{}
{May\phantom{.~0}3, 2011}
{May\phantom.~21, 2014}
{}


\ACMCCS{[{\bf Theory of computation}]: Logic; Computational complexity
  and cryptography---Complexity classes; [{\bf Computing
      methodologies}]: Artificial intelligence---Knowledge
  representation and reasoning}

\usepackage[latin2]{inputenc}
\usepackage{amsmath}
\usepackage{amssymb}
\usepackage{amsthm}
\usepackage{hyperref}

\usepackage{xspace}

\usepackage[all]{xy}

\usepackage{color}

\theoremstyle{plain}

\newcommand{\calC}{\mathcal{C}}
\newcommand{\calD}{\mathcal{D}}

\newcommand{\calH}{\mathcal{H}}
\newcommand{\calI}{\mathcal{I}}

\newcommand{\calK}{\mathcal{K}}
\newcommand{\calL}{\mathcal{L}}

\newcommand{\calO}{\mathcal{O}}
\newcommand{\calP}{\mathcal{P}}
\newcommand{\calQ}{\mathcal{Q}}
\newcommand{\calR}{\mathcal{R}}

\newcommand{\calT}{\mathcal{T}}

\newcommand{\frakA}{\mathfrak{A}}
\newcommand{\frakB}{\mathfrak{B}}
\newcommand{\frakC}{\mathfrak{C}}

\newcommand{\frakF}{\mathfrak{F}}

\newcommand{\frakH}{\mathfrak{H}}
\newcommand{\frakI}{\mathfrak{I}}

\newcommand{\frakM}{\mathfrak{M}}

\newcommand{\frakQ}{\mathfrak{Q}}
\newcommand{\frakR}{\mathfrak{R}}
\newcommand{\frakS}{\mathfrak{S}}
\newcommand{\frakT}{\mathfrak{T}}

\newcommand{\NN}{\mathbb{N}}

\newcommand{\liff}{\leftrightarrow}
\newcommand{\limp}{\rightarrow}

\newcommand{\mathsc}[1]{\text{\textsc{#1}}} 

\newcommand{\PTime}{\ensuremath{\mathsc{Ptime}}\xspace}
\newcommand{\ExpTime}{\ensuremath{\mathsc{ExpTime}}\xspace}
\newcommand{\NExpTime}{\ensuremath{\mathsc{NExpTime}}\xspace}
\newcommand{\twoExpTime}{\ensuremath{\mathsc{2ExpTime}}\xspace}

\newcommand{\FOtwo}{\ensuremath{\mathsf{FO}^2}\xspace}
\newcommand{\Ctwo}{\ensuremath{\mathsf{C}^2}\xspace}
\newcommand{\GCtwo}{\ensuremath{\mathsf{GC}^2}\xspace}
\newcommand{\GF}{\ensuremath{\mathsf{GF}}\xspace}
\newcommand{\GFO}{\ensuremath{\mathsf{GF}}\xspace}
\newcommand{\CGF}{\ensuremath{\mathsf{CGF}}\xspace}

\newcommand{\LFP}{\ensuremath{\mathsf{LFP}}\xspace}
\newcommand{\IFP}{\ensuremath{\mathsf{IFP}}\xspace}

\newcommand{\CQ}{\ensuremath{\mathsf{CQ}}\xspace}
\newcommand{\BCQ}{\ensuremath{\mathsf{BCQ}}\xspace}
\newcommand{\UCQ}{\ensuremath{\mathsf{UCQ}}\xspace}
\newcommand{\ACQ}{\ensuremath{\mathsf{ACQ}}\xspace}

\newcommand{\atype}{{\mathrm{atp}}}

\newcommand{\gbisim}{\sim_\mathrm{g}}
\newcommand{\gbisimn}[1]{\sim_\mathrm{g}^{#1}}

\newcommand{\downcl}[1]{{#1}\!{\downarrow}}

\newcommand{\covers}{\stackrel{\sim~}{\rightarrow}}
\newcommand{\longcovers}{\stackrel{\sim}{\longrightarrow}}

\newcommand{\Inv}{{\mathbb{I}}}
\newcommand{\Gam}{{\mathbb{G}}}

\newcommand{\GamDom}{{G}}

\newcommand{\canon}{{\mathrm{can}}}

\newcommand{\dom}{{\mathrm{dom}}}
\newcommand{\img}{{\mathrm{img}}}
\newcommand{\rhoext}{\hspace{0.1em}{\buildrel \rho \over \longrightarrow}\hspace{0.1em}}
\newcommand{\fatrho}{\rho\hspace{-0.565em}\rho\hspace{0.1em}}

\newcommand{\tup}[1]{{\overline{{#1}}}}  

\newcommand{\sib}{\equiv} 
\newcommand{\pred}{\prec} 
\newcommand{\predinv}{\succ}   

\begin{document}


\newcommand{\nop}[1]{}             

\title[Querying the Guarded Fragment]{Querying the Guarded Fragment\rsuper*}

\author[V.~B\'{a}r\'{a}ny]{Vince B\'{a}r\'{a}ny\rsuper a}
\address{{\lsuper{a,c}}Department of Mathematics, Technische Universit\"at Darmstadt} 
\email{\{vbarany,otto\}@mathematik.tu-darmstadt.de} 
\thanks{{\lsuper a}Vince B\'{a}r\'{a}ny gratefully acknowledges funding received 
from ERC Starting Grant Sosna while affiliated with the Institute of Informatics 
at Warsaw University.}

\author[G.~Gottlob]{Georg Gottlob\rsuper b}
\address{{\lsuper b}Oxford University Computing Laboratory, Wolfson Building, Parks Rd., OX1 3QD Oxford, UK}
\email{georg.gottlob@comlab.ox.ac.uk}
\thanks{{\lsuper b}Georg Gottlob's research was funded by the European Research Council under the European 
  Community's Seventh Framework Programme (FP7/2007-2013) / ERC grant agreement no.~246858. 
  Georg Gottlob also gratefully acknowledges a Royal Society Wolfson Research Merit Award.}

\author[M.~Otto]{Martin Otto\rsuper c}
\address{\vspace{-18 pt}}
\thanks{{\lsuper c}Martin Otto's research has been partially supported by DFG grant no.~OT~147/5-1.}




\titlecomment{{\lsuper*}This is an improved and extended version of~\cite{BGOLics}.}


\begin{abstract}
Evaluating a Boolean conjunctive query $q$ against a guarded first-order theory 
$\varphi$ is equivalent to checking whether $\varphi\wedge \neg q$ is unsatisfiable. 
This problem is relevant to the areas of database theory and description logic. 
Since $q$ may not be guarded, well known results about the decidability, complexity, 
and finite-model property of the guarded fragment do not obviously carry over to 
conjunctive query answering over guarded theories, and had been left open in general. 
By investigating finite guarded bisimilar covers of hypergraphs and relational structures, 
and by substantially generalising Rosati's finite chase, we prove for guarded 
theories $\varphi$ and (unions of) conjunctive queries $q$ that 
(i) $\varphi\models q$ iff  $\varphi \models_{\mathrm fin} q$, 
  that is, iff $q$ is true in every finite model of $\varphi$ and 
(ii) determining whether $\varphi\models q$ is {2EXPTIME}-complete. 
We further show the following results:   
(iii) the existence of polynomial-size conformal covers 
  of arbitrary hypergraphs;  
(iv) a new proof of the finite model property of the 
  clique-guarded fragment; 
(v) the small model property of the guarded fragment 
  with optimal bounds;  
(vi) a polynomial-time solution to the canonisation problem 
  modulo guarded bisimulation, which yields 
(vii) a capturing result for guarded bisimulation invariant 
  PTIME.
\end{abstract}


\maketitle


\section{Introduction} \label{sec:intro}

The {\em guarded fragment} of first-order logic (\GFO), defined through 
the relativisation of quantifiers by atomic formulas, was introduced 
by Andr\'eka, van Benthem, and N\'emeti~\cite{ABN98JPL}, who proved that 
the satisfiability problem for $\GF$ is decidable. 
Gr\"adel~\cite{Gr99JSL} proved that every satisfiable guarded first-order sentence 
has a finite model, i.e., that $\GF$ has the {\em finite model property} (FMP). 
In the same paper, Gr\"adel also proved that satisfiability of $\GFO$-sentences 
is complete for \twoExpTime, and \ExpTime-complete for sentences involving 
relations of bounded arity. 
The guarded fragment has since been intensively studied and extended in various ways. 
For example, the {\em clique guarded fragment $(\CGF)$}~\cite{Gr99Clique} properly 
extends \GFO but still enjoys the finite model property as shown by Hodkinson~\cite{Hod02SL}, 
see also~\cite{HO03} for a simpler proof. 
Guardedness has emerged as a main new paradigm for decidability and other benign properties 
such as the FMP, and has applications in various areas of computer science.  
While the guarded fragment was originally introduced to embed and naturally extend 
propositional modal logics within first-order logic~\cite{ABN98JPL}, it has various 
applications and was more recently shown to be relevant to description logics~\cite{Gr98dl} 
and to database theory~\cite{Ros06,CGK-KR08}. 
Fragments of \GFO were recently studied for query answering  in such contexts, 
see e.g.~\cite{CGK-KR08,CGL-PODS09,CGL-ICDT09,Ros06,Calva06,Dllite,PMH09}. 
The main problems studied in the present paper are motivated by such applications.

\subsection{Main problems studied} 

In the present paper we study the problem of {\em querying guarded theories} 
using conjunctive queries or unions of conjunctive queries. 
A Boolean conjunctive query (\BCQ) $q$ consists of an existentially closed 
conjunction of atoms. A union of Boolean conjunctive queries (\UCQ) is 
a disjunction of a finite number of \BCQ.
If $\varphi$ is a guarded sentence (or, equivalently, a finite guarded theory), 
we say that a query $q$ evaluates to true against $\varphi$, iff $\varphi \models q$.   
In this context, we consider the following non-trivial main questions.

\paragraph*{\emph{Finite controllability.}} 
Is it true that for each $\GFO$-sentence $\varphi$ and each \UCQ $q$ 
if all finite models of $\varphi$ satisfy $q$ then the same is true 
also for all infinite models of $\varphi$, in symbols, that 
$\varphi \models q \,\iff \, \varphi \models_\mathrm{fin} q$ ?
Since the query $q$ may not be guarded, the finite model property 
of the guarded fragment is not sufficient to answer this question
positively. Rather, this question amounts to whether for each
$\varphi$ and $q$ as above, whenever $\varphi\wedge \neg q$ is
consistent, it also has a finite model. 
This is equivalent to the finite model property of the extended 
fragment $\GFO^+$ of $\GFO$, where universally quantified Boolean
combinations of negative atoms can be conjoined to guarded sentences.
The concept of finite controllability was introduced by
Rosati~\cite{Ros06,Ros11}.\footnote{Rosati's definition is slightly stronger 
(see Proposition~\ref{prop_Rosati}); the definition given here is, 
however, better suited for the full guarded fragment.}

\paragraph*{\emph{Size of finite models.}}
How can we bound the size of finite models? In particular, 
in case $\varphi\not\models q$, how can we bound the size of the smallest  
finite models $\frakM$ of $\varphi$ for which $\frakM \models \neg q$\,? 
Note that any recursive bound on the size of such models $\frakM$ 
immediately yields the decidability of query answering. 
On the other hand, if $\varphi$ is consistent and  $\varphi\models q$, 
then the existence of a finite model $\frakM$ such that $\frakM\models q$ 
follows trivially from the FMP of $\GFO$, 
because every model $\frakM$ of $\varphi$ is also a model of $q$. 
However, little was known about the size of the smallest finite models
of a satisfiable guarded sentence  $\varphi$. Gr\"adel's finite-model
construction in~\cite{Gr99JSL}, in case of unbounded arities, 
first transforms $\varphi$ into a doubly exponentially sized structure, 
which is then input to a transformation according to Herwig's theorem~\cite{Herwig}, 
requiring a further exponential blow-up in the worst case.
This suggests a triple-exponential upper bound. Can we do better? 

\paragraph*{\emph{Hypergraph covers.}}
Approaching the above problems on a slightly more abstract level we construct 
hypergraph covers satisfying certain acyclicity criteria, which we refer to 
as weakly $N$-acyclic covers (see Section~\ref{sec_covers}).  
Very informally, a hypergraph cover is a bisimilar companion $\hat{\frakH}$ 
of $\frakH$ with a bisimulation induced by a homomorphic projection 
$\pi: \hat{\frakH} \to \frakH$. And weak $N$-acyclicity implies that every 
subset of $\hat{\frakH}$ of size at most $N$ projects via $\pi$ into an acyclic 
(not necessarily induced) sub-hypergraph of $\frakH$. 
This notion is thus intimately related to the query answering problem and 
the existence of finite covers is a key to finite controllability. 
This problem has been previously studied by the third author, who 
in~\cite{Otto09rep} gave a non-elementary construction of weakly $N$-acyclic 
hypergraph covers. 

Of further interest, in particular in connection with finite controllability 
of query answering for the more general clique-guarded fragment, is the 
existence of hypergraph covers that are both conformal and weakly $N$-acyclic
for a suitable $N$. 
Existence of conformal covers, with no regard to acycliciy constraints, was 
established in~\cite{HO03}; their doubly exponential construction being the 
only known bound.

Is it possible to find better, possibly polynomial constructions of hypergraph 
covers with the above properties?

\paragraph*{\emph{Decidability and complexity.}}
Is \UCQ-answering over guarded theories decidable, and if so, 
what is the complexity of deciding whether $\varphi\models q$ 
for a \UCQ $q$ and a guarded sentence $\varphi$?

\paragraph*{\emph{Canonisation and capturing.}}
A further problem of independent interest, which can be solved on the basis 
of the methods developed for the above questions, is \PTime canonisation 
--- the problem of providing a unique representative for each guarded 
bisimulation equivalence class of structures, to be computed in \PTime 
from any given member of that class. 
This has implications for capturing the guarded bisimulation invariant 
fragment of \PTime in the sense of descriptive complexity.\\

We provide answers to all these questions. Before summarising our results, 
let us briefly explain how the above questions relate to database theory 
and description logic.

\subsection{Applications to databases and description logic}

In the database area, {\em query answering under integrity constraints} 
plays an important role. In this context a relational database $D$, 
consisting of a finite set (conjunction) of ground atoms
is given, and a set $\Sigma$ of {\em integrity constraints} is
specified on $D$. The database $D$ does not necessarily satisfy 
$\Sigma$, and may thus be ``incomplete''. The problem of answering a
\BCQ $q$ on $D$ under $\Sigma$ consists of determining whether
$D\cup\Sigma\models q$, also written as $(D,\Sigma)\models q$. 
 
An important class of integrity constraints in this context are 
the so-called tuple-generating dependencies~\cite{BV84}.
Given a relational schema (i.e., signature) $\calR$, a
\emph{tuple-generating dependency} (TGD) $\sigma$ over $\calR$ 
is a first-order formula of the form
$\forall \tup{x}\forall \tup{y} \bigl(\Phi(\tup{x},\tup{y}) 
 \limp \exists{\tup{z}} \,\Psi(\tup{x},\tup{z})\bigr)$, 
where $\Phi(\tup{x},\tup{y})$ and $\Psi(\tup{x},$ $\tup{z})$ 
are conjunctions of atoms over~$\calR$, called the \emph{body} 
and the \emph{head} of $\sigma$, respectively.
It is well known that database query answering under TGD is undecidable, 
see~\cite{BV81}, even for very restricted cases~\cite{CGK-KR08}.
For the relevant class of {\em guarded TGD}~\cite{CGK-KR08}, however,
query answering is decidable and actually {\sc 2ExpTime}-complete~\cite{CGK-KR08}.
A TGD $\sigma$ is \emph{guarded} (GTGD) if it has an atom in its body that contains 
all universally quantified variables of $\sigma$. For example, the sentence
$$
\begin{array}{@{}c@{}}
  \forall \mathit{M,N,D}\  
\bigl( 
  (
  \mathit{Emp(M,N,D)} \land 
  \mathit{Manages(M,D)} 
  )
  \limp \hfill\ \ \ \phantom{mmmmmm} \\ 
  \phantom{mm.}\exists \mathit{E,  N'}\  
  (
  Emp(E,N',D) \land \mathit{Reportsto(E,M)}
  )
\bigr)
\end{array}
$$
is a GTGD stating that if  $M$ is a manager named $N$ belonging to and
managing department $D$, then there must be at least one employee $E$
having some name $N'$ in department $D$ reporting to $M$. 
In general, GTGD are, strictly speaking, not  guarded sentences, because 
their heads may be unguarded. However, by using ``harmless'' auxiliary 
predicates and splitting up TGD heads into several rules, each set of GTGD 
can be rewritten into a guarded sentence that is (for all relevant purposes) 
equivalent to the original set. For instance, the above TGD can be rewritten 
into the following three guarded sentences
$$
\begin{array}{l}
\forall \mathit{M,N,D}\  
\bigl(
    (
    \mathit{Emp(M,N,D)} \land \mathit{Manages(M,D)} 
    )
    \limp \exists \mathit{E,N'}\ \mathit{aux(M,D,E,N')}
\bigr)
;\\
\forall\mathit{M,D,E,N'}\  
\bigl(
\mathit{aux(M,D,E,N')}\, \limp\, \mathit{Emp(E,N',D)}
\bigr)
; \\
\forall\mathit{M,D,E,N'}\  
\bigl(
\mathit{aux(M,D,E,N')}\, \limp\, \mathit{Reportsto(E,M)}
\bigr)
.
\end{array}
$$

The class of {\em inclusion dependencies} (ID) is a simple subclass 
of the class of GTGD. An ID has the logical  form 
$
  \forall\,\tup{x},\tup{y}\,(\alpha(\tup{x},\tup{y}) 
      \limp \exists\tup{z}\,\beta(\tup{x},\tup{z}))
$, 
where $\alpha$ and $\beta$ are single atoms.  
In~\cite{JK84} it was shown that query answering under ID is decidable and, 
more precisely, {\sc PSpace}-complete in the general case and {\sc NP}-complete 
for bounded arities. One very important problem was left open in~\cite{JK84}:
the finite controllability of query answering in the presence of IDs. 
Given that in the database world attention is limited to {\em finite databases}, 
a Boolean query that would be false in infinite models of $D\cup\Sigma$ only, 
would still be finitely satisfied by $D\cup\Sigma$ and should be answered 
positively. Do such queries exist? This problem was solved by Rosati~\cite{Ros06}, 
who, by using a finite model generation procedure called {\em finite chase}, 
showed that query answering in the presence of IDs is finitely controllable. 
Rosati's result is actually formulated as follows.

\begin{prop}[Rosati \cite{Ros11}] \label{prop_Rosati}
For every finite set of facts $D$ and set $\calI$ of ID and for every $N$
there exists a finite structure $\frakC$ extending $D$ and satisfying $\calI$
and such that for every Boolean conjunctive query $q$ comprised of at most $N$ 
atoms $\frakC \models q$ iff $D,\calI \models q$.  
\end{prop}

\emph{Description logics} are used for ontological reasoning in the Semantic Web 
and in other contexts. Some description logics such as DL-Lite$_{\mathit core}$ 
and DL-Lite$_{\calR}$ \ \cite{Dllite} are essentially based on IDs, and are 
thus finitely controllable. 
The already mentioned class of GTGD and the yet more expressive class of 
{\em weakly guarded TGD (WGTGD)} have been introduced and studied 
in~\cite{CGK-KR08,CGL-PODS09} as powerful tools for data integration, 
data exchange~\cite{FKMP05}, and ontological reasoning. 
As shown in~\cite{CGL-PODS09}, the class of GTGD augmented to also allow rules 
with the truth constant $\bot$ (= "false") as their head, generalizes the 
main DL-Lite description logic families.  
The WGTGD class is yet more general, and captures, unlike the GTGD class, plain Datalog. 
The finite controllability of GTD and WGTD theories, however, was left as an open problem. 
Unfortunately, Rosati's finite chase cannot be directly applied to GTGD or to WGTGD. 

Let us briefly sketch how finite controllability of query answering 
in the presence of GTGD and WGTGD follows from the finite controllability 
of query answering against \GFO, which is the main result of the present paper. 
For GTGD theories this is easy. As explained above, they can be rewritten 
as guarded sentences and can thus be considered a sub-fragment of \GFO.
Let us now turn our attention to WGTGD theories, and first give some intuition 
about how they are defined. 
Roughly, for a TGD set $\Sigma$, the set of all argument positions $\Pi$ 
of all atoms of predicates of $\Sigma$ can be partitioned into sets $\Pi_A$ 
and $\Pi_U$. $\Pi_A$ are the so-called {\em affected positions}, where, 
when the rules are executed over a database $D$, Skolem terms (i.e., new instance 
values of existentially quantified head-variables) may need to be introduced, 
whereas $P_U$ are those argument positions, which never need to hold Skolem terms 
(see~\cite{CGK-KR08} for a more precise definition).
A TGD set $\Sigma$ is {\em weakly guarded} (i.e., $\Sigma$ is a WGTGD set) 
if each rule body has an atom (a {\em weak guard}) that covers all those 
body variables that only occur in affected positions. 
Note that weakly guarded TGD sets are, in general, unguarded. 
However, each theory $(D,\Sigma)$, where $D$ is a database and $\Sigma$ a WGTDG set 
can be replaced by an equivalent theory $(D,\Sigma')$ where $\Sigma'$ is a GTGD set 
as follows: $\Sigma'$ is obtained from $\Sigma$ by replacing each rule $\sigma$ 
of $\Sigma$ by all possible instantiations of variables in unaffected positions 
in $\Sigma$ by constants from $D$. It is easy to see that for each \UCQ $q$, 
$(D,\Sigma)\models q$ iff  $(D,\Sigma')\models q$ and, moreover, 
if $(D,\Sigma')\not\models q$, then for each model $M$ of $D$ and $\Sigma'$ 
such that $M\not\models q$ it also holds that $M\models \Sigma$. 
It follows that query answering under WGTGD theories is finitely controllable 
if query answering under GTGD theories is finitely controllable. Thus, if we can 
establish that query answering under \GFO theories is finitely controllable, 
then query answering under GTGD theories (constituting a sub-fragment of \GFO) 
is finitely controllable, and so is query answering under WGTGD theories.

\subsection{Summary of results}

$ $

\paragraph*{\emph{Finite Controllability.}}
That answering \UCQ against guarded sentences is finitely controllable 
was already implicit in the report~\cite{Otto09rep}, although not formulated 
in this terminology. The finite models constructed in~\cite{Otto09rep} 
are of non-elementary size and do not yield meaningful complexity results. 
The following central result of our paper, derived by a completely new proof, 
yields a much better size bound.

\begin{thm} \label{thrm_FinControl}
For every \GFO sentence $\varphi$ and every \UCQ $q$, 
$\;\varphi \models q   \iff   \varphi \models_{\mathrm{fin}} q$. 
More specifically, if $\varphi \land \lnot q$ is satisfiable 
then it has a model of size $2^{|\varphi| {|q|}^{c |q|^2}}$,
where $c$ depends solely on the signature of $\varphi$. 
\end{thm}

\begin{cor}
Answering \UCQ against GTGD or WGTGD theories is finitely controllable.
\end{cor}

More refined estimates on the size of finite models are provided in 
Section~\ref{sec_FMP}. To obtain Theorem~\ref{thrm_FinControl}, we 
establish new results on hypergraph covers, which are of independent~interest.

\paragraph*{\emph{Hypergraph Covers.}}
We relate finite controllability to the concept of hypergraph covers. 
A \emph{hypergraph cover} for a given hypergraph $\frakA$ consists of 
a hypergraph $\frakB$ together with a homomorphism $\pi \colon \frak{B} \covers \frakA$ 
that induces a hypergraph bisimulation between $\frakB$ and $\frakA$. 
This notion naturally extends to relational structures $\frakA, \frakB$
on the basis of homomorphism-induced guarded bisimulations. The 
following main technical result is used to derive most other results
(for definitions of notions mentioned see Section 2).

\begin{thm}[Main Technical Result] \label{thrm_main}
Given $N \geq 2$ and a hypergraph $\frakA$, one can construct an $N$-conformal 
hypergraph $\frakA^{(N)}$ constituting a weakly $N$-acyclic hypergraph cover of $\frakA$. 
Moreover, $|\frakA^{(N)}| = |\frakA|^{w^{\calO(N^2)}}$ and 
$\frakA^{(N)}$ is conformal whenever $N>w$, where $w$ is the width of $\frakA$ and, 
for fixed $w$ and $N$, $\frakA^{(N)}$ can be computed in polynomial time from $\frakA$. 
\end{thm}

The direct analogue for guarded covers of relational structures is immediate. 
In both settings we call $\frakA^{(N)}$ the {\em $N$-th Rosati cover} of $\frakA$.

Let us explain very informally the role of the Rosati covers 
in establishing Theorem~\ref{thrm_FinControl}. 
In an easy but key step (Lemma~\ref{lemma_treeification} in Section~\ref{sec_prelim}) 
we first reduce a problem instance $\varphi\models q$ for a \GF-sentence $\varphi$ 
and a \UCQ $q$ to the equivalent question of entailment $\varphi \models \chi_q$, 
where $\chi_q$ is a disjunction of \emph{acyclic queries} stemming from the original 
query $q$. Crucially, being acyclic, $\chi_q$ can be equivalently reformulated as 
a \GF-sentence, ultimately reducing the initial query answering problem to the 
unsatisfiability of the \GF-sentence $\varphi \land \lnot \chi_q$.
It is more difficult to show that this reduction is also valid over finite models, 
i.e., that $\varphi \models_\mathrm{fin} q \,\iff\, \varphi \models_\mathrm{fin} \chi_q$.
In particular, that given a finite $\frakA \models \varphi \land \lnot \chi_q$, 
a finite model of $\varphi \land \lnot q$ can also be found.
Observe that the ``unravelling'' of $\frakA$ yields a tree-like model $\frakA^\ast$ 
of $\varphi \land \lnot \chi_q$ and an acyclic cover of $\frakA$. 
Thus, by virtue of acyclicity, $\frakA^\ast \models \lnot q$. 
However, $\frakA^\ast$ is typically infinite. 
The challenge is to find a finite cover of $\frakA$ retaining 
a ``sufficient degree of acyclicity'' so as not to render it a model of $q$. 
This is captured by the notion of a weakly $N$-acyclic cover $\frakA^{(N)}$ of $\frakA$,
which ensures that, similarly to tree unravellings, 
$\frakA^{(N)} \models q$ implies $\frakA \models \chi_q$, 
but with the qualification that $|q| \leq N$.
Theorem~\ref{thrm_main} shows that such covers can be constructed.

\paragraph*{\emph{Conformal covers.}}
Hodkinson and Otto showed in~\cite{HO03} that all hypergraphs admit guarded bisimilar 
covers by conformal hypergraphs (for definitions, see Section~\ref{sec_prelim}). 
While the construction in~\cite{HO03} involves a doubly exponential blow-up in size, 
we here obtain a polynomial construction of conformal covers as a corollary 
to Theorem~\ref{thrm_main}. 

\begin{cor} 
\label{coroll_conformal covers}
Every hypergraph $\frakH$ of width $w$ admits a conformal 
hypergraph cover of size $|\frakH|^{w^{\calO(w)}}$. 
For bounded width, we  thus obtain polynomial size conformal covers.
\end{cor}

\paragraph*{\emph{Finite model property of the clique-guarded fragment.}}
As it happens, our construction used for Theorem~\ref{thrm_FinControl} also yields 
an extension of Theorem~\ref{thrm_FinControl} to the clique-guarded fragment, \CGF.

\begin{thm} \label{thrm_FinControl_CGF}
For every $\varphi \in \CGF$ and every $q \in \UCQ$ we have  
$\; \varphi \models q   \iff   \varphi \models_{\mathrm{fin}} q$. 
More specifically, if $\varphi \land \lnot\, q$ is satisfiable 
then it has a finite model of size $2^{(|\varphi|+|\tau|^{\calO(h)}) (wh)^{\calO(wh^2)}}$,
where $h$ is the height of $q$, $\tau$ is the signature of $\varphi$, 
and $w = \mathrm{max}\{ \mathrm{width}(\varphi), \mathrm{width}(\tau) \}$.
\end{thm}

In particular, we obtain finite models of any satisfiable clique-guarded formula,
and thereby a new proof of the Finite Model Property of the clique-guarded fragment.
In fact, our construction yields more compact finite models than hitherto known.

\paragraph*{\emph{Small model property.}}
Through our new method of finite-model construction, we are able to improve 
the bounds implicit in~\cite{Gr99JSL} for \GFO and the overhead for \CGF 
implicit in~\cite{Hod02SL,HO03} on the size of the smallest finite model 
of a satisfiable (clique-)guarded sentence.

\begin{thm} \label{thrm_small_models}
Every satisfiable formula of \CGF $($and thus of $\GFO)$ has a finite model of size 
exponential in the length and doubly exponential in the width of the formula.
Moreover, for every $k \geq 2$, the $k$-variable fragment of \CGF $(\GFO)$
has finite models of exponential size in the length of the formula.
\end{thm}

Another important fragment of first-order logic that has the finite model property 
is $2$-variable first-order logic, denoted by $\FOtwo$. 
It was shown in~\cite{GKV97} that if an $\FOtwo$ formula $\phi$ is satisfiable, 
then it has a model of cardinality singly exponential in the size of $\phi$, 
improving an earlier doubly exponential bound by Mortimer~\cite{Mortimer}.
As a consequence, \cite{GKV97} also proved \NExpTime-completeness of the 
satisfiability problem for $\FOtwo$. 
A more powerful fragment that embeds a number of key features of description 
logics is the extension of the $2$-variable fragment with counting quantifiers, 
denoted $\Ctwo$. In contrast to \FOtwo, \Ctwo does not have the finite model property, 
but computationally it is no more difficult than \FOtwo: 
both satisfiability and finite satisfiability of \Ctwo-formulas are decidable 
and \NExpTime-complete~\cite{GOR97,PSzT99,PH05}.

In~\cite{Gr98dl} Gr\"adel proposed the guarded fragment of $\FOtwo$ as a testbed  
for simple description logics. A more suitable fragment for this purpose is obtained 
by imposing on the one hand a restriction to guarded quantification while allowing 
on the other hand the use of counting quantifiers. The two-variable guarded fragment 
with counting quantifiers, \GCtwo, properly subsumes the description logic $\mathcal{ALCQI}$, 
cf.~e.g.~\cite{Calva96}, for which finite satisfiability was shown decidable in \ExpTime 
by Lutz et al.~\cite{LST05}. In~\cite{Kaz04} Kazakov gave a polynomial, 
satisfiability-preserving translation from \GCtwo to the $3$-variable guarded fragment $\GFO^3$ 
thus establishing \ExpTime-completenes of satisfiability of \GCtwo-formulas. 
Finite satisfiability for \GCtwo was also shown to be \ExpTime-complete 
by Pratt-Hartmann~\cite{PH07}. The latter decision method is based on a 
reduction to integer programming. It is interesting to note that the optimal 
lower bound on the size of smallest finite models of finitely satisfiable 
\GCtwo-formulas is doubly exponential~\cite{PH07}.

\paragraph*{\emph{Complexity of query answering.}}
In \cite{Gr99JSL} Gr\"adel proved that satisfiability of \GFO-sentences 
is complete for \twoExpTime, and \ExpTime-complete in case of bounded arity.
We show that, more generally, answering \UCQ on the class of models of a 
\GFO-sentence can also be performed in \twoExpTime, which solves the initially 
posed complexity question about query answering over guarded theories. 
It follows from the work of Lutz \cite{Lutz07}, however, that query answering 
remains \twoExpTime-complete for \BCQ even in the bounded arity case.
Considering unions of \emph{acyclic} conjunctive queries we derive an \ExpTime 
solution to the query answering problem, and prove \ExpTime-completeness 
already for a particular fixed \GFO sentence. 

Our algorithms are built around Gr\"adel's solution of the satisfiability 
problem for \GFO. The first step consists in reducing $\varphi\models q$ 
to $\varphi\models \chi_q$, where $\chi_q$ is a disjunction of acyclic queries 
stemming from the original query. 
The formula $\chi_q$ may, however, be of exponential size in terms of 
the length of $q$, demanding a closer inspection of the contribution of 
different dimensions of $\chi_q$ to the overall complexity of checking 
(un)satisfiability of the guarded sentence $\varphi\wedge\neg\chi_q$.

We also investigate the problem of query answering over models of a fixed 
guarded sentence, and provide a number of useful bounds. Our bounds for 
fixed sentences $\varphi$ are not all tight and leave room for future research.

\paragraph*{\emph{Canonisation and capturing.}}
As a further consequence of the proof method employed for Theorem~\ref{thrm_main}, 
we find a polynomial solution to the canonisation problem for guarded bisimulation 
equivalence $\gbisim$. This allows us to capture the $\gbisim$-invariant fragment 
of \PTime in the sense of descriptive complexity, i.e., to provide effective 
syntax with \PTime model checking for the class of all \PTime queries
that are closed under guarded bisimulation equivalence.
Canonisation is achieved through inversion of the natural game invariant 
$\Inv(\frakA)$ that uniquely characterises the guarded bisimulation class, 
or the complete $\GF$-theory, of a given finite structure $\frakA$. 
A \PTime reconstruction of a model from the abstract specification 
of its equivalence class yields \PTime canonisation.

\begin{thm} \label{thrm_canon}
For every relational signature $\tau$ there exists a \PTime algorithm
computing from a given invariant $\Inv(\frakA)$ of an unspecified 
$\tau$-structure $\frakA$ a finite $\tau$-structure $\canon(\frakA)$ 
such that $\Inv(\canon(\frakA)) = \Inv(\frakA)$; hence $\canon(\frakA) \gbisim \frakA$ 
and $\canon(\frakA') = \canon(\frakA)$ whenever $\frakA \gbisim \frakA'$. 
\end{thm}

\begin{cor} \label{coroll_capture}
The class of all those \PTime Boolean queries that are invariant 
under guarded bisimulation, $\mbox{\sc Ptime}/{\gbisim}$, 
can be captured in the sense of descriptive complexity. 
\end{cor}

\subsection*{Organisation}

Section~\ref{sec_prelim} defines the main concepts and introduces guarded bisimilar hypergraph 
covers as a main tool. It  also states the above-mentioned  Lemma~\ref{lemma_treeification}.
Section~\ref{sec_Rosati} presents the construction of the Rosati cover. 
From this and Lemma~\ref{lemma_treeification}, the finite controllability of $\GFO$ is proven 
in Section~\ref{sec_FMP}. Section~\ref{sec_Complexity} establishes our new complexity results.  
Section~\ref{sec_Capture} deals with canonisation and capturing.


\section{Hypergraphs and guarded fragments} \label{sec_prelim}


We work with finite relational signatures. 
Let us fix such a signature $\tau$ and let $\mathrm{width}(\tau)$ denote the
the maximal arity of any of the relation symbols in $\tau$.

\subsection{Guarded fragments}

The guarded fragment of first-order logic, $\GF$, as introduced by 
Andr{\'e}ka et al.~\cite{ABN98JPL}, is the collection of first-order 
formulas with certain syntactic restrictions in the quantification pattern, 
which is analogous to the relativised nature of modal logic. 
The set of $\GF(\tau)$ formulas is the smallest set
\begin{itemize}
\item containing all atomic formulas of signature $\tau$ and equalities between variables;
\item closed under Boolean connectives: $\lnot$, $\land$, $\lor$, $\limp$, $\liff$;
\item and such that whenever $\psi(\tup{x},\tup{y})$ 
      is a $\GF(\tau)$ formula with all free variables indicated 
      and $\alpha(\tup{x},\tup{y})$ is a $\tau$-atom (or an equality)
      involving all free variables of $\psi$,  
      then the following are in $\GF(\tau)$ as well:
$$
\mbox{}\hspace{-.5cm} (\forall\, \tup{x}.\, \alpha)\, \psi := \forall \tup{x} (\alpha \limp \psi) 
  \ \ \text{and} \ \  
  (\exists\, \tup{x}.\, \alpha)\, \psi := \exists \tup{x} (\alpha \land \psi).
$$
\end{itemize}

\noindent In a $\tau$-structure $\frakA$ a non-empty set $X$ of elements is said to be \emph{guarded} 
if it is a singleton or there is an atom $R^\frakA(\tup{a})$ such that every member of $X$ 
occurs in $\tup{a}$.
A \emph{maximal guarded set} is one not properly included in any other guarded set.
A tuple $\tup{b}$ of elements is guarded if the set of its components is guarded. \\

While $\GF$ provides an important extension of the modal fragment,  
guarded quantification is too restrictive to express some basic temporal operators.
To remedy this shortcoming various relaxations of the notion of guardedness 
and corresponding fragments have been introduced, chief among them 
the clique-guarded fragment.

The \emph{clique guarded fragment}, $\CGF$, relaxes the constraints on
guards $\alpha$ in $\GF$ to allow existentially quantified conjunctions of atoms 
as guards that guarantee that the tuple of free variables is clique-guarded.
A set $X$ of elements of a structure $\frakA$ is \emph{clique-guarded}
if every pair of elements of $X$ is guarded, equivalently, 
if $X$ induces a clique in the Gaifman graph of $\frakA$.
A tuple $\tup{a}$ is clique-guarded whenever the set of its components is. 
Observe that while guarded sets are bounded in size by the width of the signature, 
there can be arbitrarily large clique-guarded sets whenever the width is at least $2$.
Recall that the \emph{width} of a formula $\varphi$, $\mathrm{width}(\varphi)$
is the maximal number of free variables in any of its subformulas.
In a clique-guarded formula $\varphi$ the maximal size of a clique-guarded set 
quantified over is bounded by $\mathrm{width}(\varphi)$.

Observe that guardedness and clique-guardedness (of tuples of any fixed arity)
are definable in the corresponding logic. 
That is, there are formulas $\mathtt{guarded}_\calL(\tup{x})$ 
expressing that the tuple $\tup{x}$ is guarded in the sense appropriate for the fragment.
E.g., $\mathtt{guarded}_\GF(\tup{x}) = \bigvee_{\alpha} \exists \tup{y} \alpha (\tup{x},\tup{y})$
where $\alpha$ ranges over all $\tau$-atoms in tuples of variables 
comprising at least $\tup{x}$, as indicated.   
A formula $\mathtt{guarded}_\CGF(\tup{x})$ can be similarly defined. \\

An \emph{atomic $\tau$-type} $t(x_1,\ldots,x_n)$ is a maximal consistent set 
of $\tau$-literals (atoms or negated atoms, including (in)equalities) 
whose constituent terms are among the variables $x_1,\ldots,x_n$ and 
the constants from $\tau$. An atomic type $t(x_1,\ldots,x_n)$ determines, 
for every choice of indices $\tup{\i} = (i_1,\ldots,i_k)$, its \emph{restriction} 
to components $\tup{\i}$, which is an atomic type in $k$ variables 
$(x_{i_1},\ldots,x_{i_k})$ denoted $t|_{\tup{\i}}$;
conversely we say that $t$ is an \emph{extension} of $t|_{\tup{\i}}$.
In a $\tau$-structure $\frakA$ the atomic type $\atype_\frakA(\tup{a})$ of a tuple $\tup{a}$ 
is the unique atomic type $t(\tup{x})$ such that $\frakA \models t(\tup{a})$. 
One says that $t$ is \emph{realised} by $\tup{a}$ in $\frakA$.
Each atomic type can be identified with the isomorphism type of the sub-structure 
induced by any tuple realising it.   
Over a signature of $r$ many relational symbols of maximal arity $w$ and $k$ constants 
there are $2^{\calO(r (n+k)^w)}$ many atomic types in $n$ variables. 
We identify each atomic type with the conjunction of its literals. 

\medskip

\paragraph*{\emph{Guarded bisimulation}} 
The notion of guarded bisimulation~\cite{Gr99JSL}, denoted $\gbisim$, 
can be defined either in terms of the guarded bisimulation game, a variant 
of the Ehrenfeucht-Fra\"iss\'e style pebble game in which the set of pebbles 
must at any given time be guarded, or as a back-and-forth system of 
partial isomorphisms whose domain and image are both guarded. 
$\GF$ is preserved under guarded bisimulation~\cite{ABN98JPL}, 
see also \cite{{Gr99JSL}, GHO02}: 
\[
\mbox{}\hspace{-.4cm}
\frakA \gbisim \frakB \quad \Longrightarrow \quad 
    \text{for all } \varphi \in \GF :\ \frakA \models \varphi \;\Leftrightarrow\; \frakB \models \varphi.
\]

Given a relational structure $\frakA$, its \emph{guarded bisimulation game graph},
denoted $\Gam(\frakA)$, has as its vertices the set $\GamDom(\frakA)$ of all 
\emph{maximal guarded tuples} of $\frakA$, each labeled by its atomic type, 
equivalently, by the isomorphism type of the substructure induced by the tuple. 
Two such tuples $\tup{a}$ and $\tup{b}$ are linked by an edge labeled by 
a partial bijection $\rho \subseteq \{1,\ldots,|\tup{a}|\}\times\{1,\ldots,|\tup{b}|\}$ 
whenever $a_i = b_j$ for all $(i,j) \in \rho$.
Note that structures $\frakA$ and $\frakB$ are guarded bisimilar 
iff $\Gam(\frakA)$ and $\Gam(\frakB)$ are bisimilar in the modal sense~\cite{GHO02}.

\medskip

The \emph{guarded bisimulation invariant} $\Inv(\frakA)$ of $\frakA$ 
is defined as the bisimulation quotient of $\Gam(\frakA)$.
Vertices of $\Inv(\frakA)$ correspond to $\gbisim$-classes of maximal 
guarded tuples of $\frakA$, labeled by their atomic types (induced isomorphism types). 
A $\rho$-labeled edge links vertices $v$ and $w$ if there are 
guarded tuples $\tup{a}$ and $\tup{b}$ in $\frakA$ 
realising the $\gbisim$-classes represented by $v$ and by $w$, respectively, 
and such that $a_i = b_j$ for all $(i,j) \in \rho$.

\paragraph*{\emph{Scott normal form and satisfiability criterion}}
Gr\"adel's analysis of decidability for $\GF$ \cite{Gr99JSL}
uses the following Scott normal form corresponding to a 
relational Skolemisation.

\begin{lem}[{\cite[Lemma 3.1]{Gr99JSL}}] \label{lemma_normalform} 
To every (clique-)guarded $\tau$-sentence $\varphi$ one can associate 
a companion (clique-)guarded $\tau\cup\sigma$-sentence 
\begin{equation} \label{eq_Scott_nf}
  \psi =  \bigwedge_j (\forall\, \tup{x}. \alpha_j) \, \vartheta_j(\tup{x}) 
          \, \land \,
          \bigwedge_i (\forall\, \tup{x}. \beta_i) \, 
                        (\exists \,\tup{y}. \gamma_i) \ \psi_i(\tup{x},\tup{y}) 
\end{equation}
such that $\psi \models \varphi$ and every \ $\frakA \models \varphi$ \ 
has a $\tau\cup\sigma$-expansion $\frakB \models \psi$.
Here $|\sigma| \leq |\varphi|$, $\mathrm{width}(\psi) = \mathrm{width}(\varphi)$ 
and the $\vartheta_j$, $\psi_i$ are quantifier-free.
\end{lem}


A guarded bisimulation game graph $G$ or, similarly, a guarded-bisimulation invariant $I$ 
is said to satisfy the formula $\psi$ in Scott normal form \eqref{eq_Scott_nf} if
\begin{itemize}
\item its vertices are labeled by atomic types in the signature of $\psi$ 
  that are guarded and that satisfy the universal conjuncts of $\psi$; and
\item for each vertex $v$ with label $t(\tup{x}\tup{z})$ 
  and for each conjunct 
  $ (\forall\, \tup{x}. \beta_i) \, (\exists \,\tup{y}. \gamma_i) \ \psi_i(\tup{x},\tup{y}) $
  of $\psi$ such that $t(\tup{x}\tup{z}) \models \beta_i(\tup{x})$
  there exists a vertex $w$ labeled with some 
  type $s(\tup{x}'\tup{y}) \models \psi_i(\tup{x}',\tup{y})$ 
  such that $s|_{\tup{x}'} = t|_\tup{x}$ and $v$ and $w$ are 
  linked by an edge labeled with the mapping $\rho: \tup{x} \to \tup{x}'$.
\end{itemize}

\begin{prop}[cf. {\cite[Lemma 3.4]{Gr99JSL}}] \label{prop_satcrit}
Let $\psi$ be the normal form of $\varphi$ as in \eqref{eq_Scott_nf}. 
Then $\varphi$ is satisfiable if, and only if, 
there exists a guarded bisimulation invariant $\frakI$ satisfying $\psi$ 
and such that vertices of $\frakI$ are labeled by distinct guarded atomic types.\\[-1.4em]
\end{prop}

\subsection{Hypergraphs, acyclicity and covers} \label{sec_covers}

A \emph{hypergraph} is a pair $H=(V,S)$ with $V$ its set of elements 
and $S \subseteq \calP(V)$ a set of subsets of $V$, which are called hyperedges. 
For a set of hyperedges $S$, let $\downcl{S}$ stand for 
the closure of $S$ under subsets.
A set $X$ of elements of $H$ is guarded if $X \in \downcl{S}$. 
The Gaifman graph $\Gamma(H)$ of $H$ is the undirected graph 
having vertex set $V$ and, as edges, all non-degenerate guarded pairs of $H$.
The maximal size of any hyperedge is referred to as the \emph{width} of $H$. 
To every $\tau$-structure $\frakA$ one associates in a natural way 
a hypergraph $H[\frakA]$ with $V$ the universe of $\frakA$ 
and $S$ the collection of maximal guarded subsets of $\frakA$. 
The width of $H[\frakA]$ is then bounded by $\mathrm{width}(\tau)$. 
The Gaifman graph of $\frakA$ is $\Gamma[\frakA] := \Gamma(H[\frakA])$.

A \emph{homomorphism} $h \colon H \rightarrow H'$ between 
hypergraphs $H = (V,S)$ and $H' = (V',S')$ is a map from $V$ to $V'$ 
such that $h(s) \in S'$ for all $s \in S$. A hypergraph homomorphism $h$ 
is \emph{rigid} if $|h(s)|=|s|$ for every hyperedge $s$. 
Every homomorphism $h \colon \frakA \rightarrow \frakA'$
between relational structures induces a hypergraph homomorphism
from $H[\frakA]$ to $H[\frakA']$.

Game graphs $G(H)$ and invariants $I(H)$ are defined similarly for hypergraphs $H$ 
as for relation structures, where instead of guarded bisimulation 
we use the natural notion of \emph{hypergraph bisimulation}. 
It is safe to think of hypergraph bisimulation as of guarded bisimulation 
stripped of all atomic relational information. 
Vertices in the game graph are maximal hyperedges each labeled with 
the isomorphism type of the sub-hypergraph induced by it, in other words, 
the label carries the information about all hyperedges lying inside a maximal 
hyperedge. Edges of the game graph connect overlapping maximal hyperedges 
and are in bijection with and are labeled by partial bijections compatible 
with the actual overlap.

A hypergraph $H$ is ($N$-)\emph{conformal} if every clique in $\Gamma(H)$ 
(of size at most $N$) is covered by a hyperedge of $H$. 
A structure $\frakA$ is ($N$-)conformal whenever $H[\frakA]$ is, i.e., 
if every $k$-clique ($k\leq N$) in its Gaifman graph is covered by a ground atom.  
Over conformal structures guarded quantification is as powerful as 
clique-guarded quantification.

A hypergraph $H$ is ($N$-)\emph{chordal} if all cycles in $\Gamma(H)$ 
of length greater than $3$ (and at most $N$) have a chord in $\Gamma(H)$. 
An analogous notion for relational structures $\frakA$ is similarly defined 
in terms of the Gaifman graph $\Gamma(H(\frakA))$.

A hypergraph is ($N$-)\emph{acyclic} if it is both ($N$-)chordal 
and ($N$-)conformal. For finite hypergraphs acyclicity is equivalent 
to tree decomposability. A finite hypergraph is \emph{tree decomposable} 
if it can be reduced to the empty hypergraph by iteratively deleting 
some non-maximal hyperedge or some vertex contained in at most one hyperedge 
(Graham's algorithm) cf.~\cite{BFMY83}. 
We say that a relational structure $\frakA$ is \emph{guarded tree decomposable}
if $\frakA$ allows a tree decomposition in the sense of Robertson--Seymour 
with guarded bags. 
This is equivalent to $H[\frakA]$ being tree decomposable, i.e.~acyclic. 

A homomorphism $h \colon H \rightarrow H'$ into the hypergraph $H'=(V',S')$
is called \emph{tree decomposable} if there is some $S'' \subseteq \downcl{S'}$ 
such that \ $h \colon H \rightarrow H''$ is a homomorphism into $H'' = (V',S'')$ 
and $H''$ is tree decomposable; anaolgously, for relational structures. 
In this sense a homomorphism between relational structures $h \colon \frakH \to \frakA$ 
is \emph{guarded tree decomposable} if $h = g \circ f$ for homomorphisms 
$f \colon \frakH \to \frakF$ and $g \colon \frakF \to \frakA$ 
where $\frakF$ is guarded tree decomposable. We describe this situation by saying 
that $h$ \emph{factors} through $f \colon \frakH \to \frakF$. 
The following key notion employs this idea to capture a subtle relaxation 
of $N$-acyclicity in the context of coverings. 

\begin{defi}[Weak $N$-acyclicity] \label{def_cover}
A \emph{guarded bisimilar cover} (or just \emph{cover} for short) 
$\pi \colon \frakB \covers \frakA$ is an onto homomorphism 
$\pi \colon \frakB \to \frakA$ inducing a guarded bisimulation
$\{ (\tup{b}, \pi(\tup{b})) \mid \tup{b} \text{ maximal}$ $\text{guarded tuple in } \frakB \}$ 
between relational structures $\frakB$ and $\frakA$ of the same vocabulary.
A cover $\pi \colon \frakB \covers \frakA$ is \emph{weakly $N$-acyclic} 
if, for all homomorphisms $h \colon \frakQ \to \frakB$ with $|\frakQ| \leq N$,
there is a guarded tree decomposable structure $\frakT$ and homomorphisms 
$f\colon\frakQ\to\frakT$, $g\colon\frakT\to\frakA$ such that $\pi\circ h = g \circ f$.
\end{defi}



Analogous notions of hypergraph covers are defined mutatis mutandis with the 
additional stipulation that the restriction of a cover homomorphism to every 
hyperedge is expressly required to be injective, i.e.~that the cover homomorphism
is to be rigid. (In the case of guarded bisimilar covers among relations structures 
the analogous condition is implied by the above definition. That is, every 
guarded bisimular cover $\pi \colon \frakB \covers \frakA$ induces a hypergraph 
cover $\hat{\pi} \colon H[\frakB] \covers H[\frakA]$ such that $\hat{\pi}$ 
is a rigid homomorphism of hypergraphs).


\subsection{Conjunctive queries}

\emph{Conjunctive queries} (\CQ) 
are formulas of the form $\exists \tup{x} \bigwedge_i \alpha_i$, 
where the $\alpha_i$ are positive literals. 
A Boolean conjunctive query (\BCQ) is one with no free variables. 
A union of (Boolean) conjunctive queries (\UCQ) is a disjunction of \BCQ.
The size $|q|$ of a \UCQ $q$ is its length as a formula, and its 
\emph{height} is the maximal size of its disjuncts (constituent \CQ).

To every \BCQ $Q = \exists \tup{x} \bigwedge_i \alpha_i$ of signature $\tau$ 
one can associate the $\tau$-structure $\calQ$ having as its universe the set of 
variables in $\tup{x}$ and atoms as prescribed by the $\alpha_i$ of $Q$. Then 
$\frakA \models Q$ iff there exists a homomorphism $h: \calQ \to \frakA$, \cite{CM77}.
We say that $Q \in \CQ$ is \emph{acyclic} 
if the associated structure $\calQ$ is acyclic. 
Note that this is equivalent to the existence of a \emph{guarded} conjunctive query 
equivalent to $Q$, i.e., one that is both in \GFO and in \CQ~\cite{GLS03}.

For each \BCQ $Q$ we define its \emph{treeification in signature $\tau$}, 
denoted $\chi^\tau_Q$, as the disjunction of all \emph{acyclic} \BCQ\ $T$ 
in the signature $\tau$ comprised of \emph{at most three times as many atoms} 
as $Q$ and such that $T \models Q$.
Further, for $q=\bigvee_i Q_i$ a \UCQ we set $\chi^\tau_q = \bigvee_i \chi_{Q_i}$. 
It is obvious that $\chi^\tau_q \models q$ for every $q$.
In the following $\tau$ will always be an expansion of the signature of $q$, 
and will be omitted whenever clear from the context or of no import.

\begin{lem} \label{lemma_treeification}
Let $\tau$ be a signature consisting of $r$ relation symbols of maximal arity $w$ 
(the width of $\tau$). Consider a \UCQ $q=\bigvee_i Q_i$ over $\tau$ and 
let $h = \mathrm{max}_i |Q_i|$ (the height of $q$).
\begin{enumerate}[label=(\roman*)]
\setlength{\itemsep}{.5ex}
\item For a $\tau$-structure $\frakA$ we have $\frakA \models \chi^\tau_q$ if 
      there is a guarded tree decomposable homomorphism 
      $\eta: \calQ_i \to \frakA$ for some $Q_i$.
\item In particular, for all \ $\varphi \in \GF[\tau]$: 
      $\varphi \models q$ iff $\varphi \models \chi^\tau_q$.
\item The size of the treeification $\chi^\tau_q$
      is at most ${r}^{\calO(h)} (hw)^{\calO(hw)}$;  \\
      moreover, $\chi^\tau_q$ can be constructed in time 
      $|q| r^{\calO(h)} (hw)^{\calO(hw)}$.
\end{enumerate}
\end{lem}

\begin{proof}
{(i)} 
  Let $\eta: \calQ_i \to \frakA$ be a guarded tree decomposable homomorphism.
  This means that $\eta: \calQ_i \to \frakB \subseteq \frakA$ for some 
  guarded tree decomposable $\frakB$, a (not-necessarily induced) 
  substructure of $\frakA$.
  Consider a fixed guarded tree decomposition of $\frakB$ represented 
  as $(V,\preceq,\gamma)$ with $(V,\preceq)$ a forest with transitive edges
  and $\gamma$ assigning to each node an atom of $\frakB$, 
  the guard of the corresponding bag of the tree decomposition.
  As $\eta$ maps $\calQ_i$ homomorphically into $\frakB$, 
  for each atom $\alpha$ of $\calQ_i$ we can pick a node $v_\alpha \in V$ 
  with $\gamma(v_\alpha)$ guarding the image of $\alpha$ under $\eta$.  
  Let $W$ be the closure of the set of all these $v_\alpha$ under 
  greatest lower bounds w.r.t.\ $\preceq$. 
  Then $(W,\preceq,\gamma|_W)$ represents a guarded tree decomposition 
  of the structure $\calT$ consisting of those atoms in the image of $\eta$ 
  together with atoms of the form $\gamma(w)$ for $w \in W$.
  Note that at least half of the nodes in $W$ are of the form $v_\alpha$, 
  therefore $\calT$ has at most three times as many atoms as $\calQ_i$.
  And we have $\eta: \calQ_i \to \calT \subseteq \frakB$.
  To $\calT$ corresponds an acyclic \BCQ $T$, whose models are 
  precisely those structures containing a homomorphic image of $\calT$.
  Then $T \models Q_i$ and hence $T$ is one of the disjuncts in $\chi^\tau_q$. 
  Therefore $\frakA \models \chi^\tau_q$. 
  
{(ii)}
  Since $\chi^\tau_q \models q$, trivially $\varphi \models \chi^\tau_q$ implies $\varphi \models q$.
  To prove the converse implication assume indirectly 
  that $\varphi \models q$ but $\varphi \land \lnot\,\chi^\tau_q$ 
  were satisfiable. 
  Note that the latter is equivalent to a guarded formula.
  Then, by the tree model property of $\GF$ \cite{Gr99JSL}, there is 
  a guarded tree decomposable model $\frakT \models \varphi \land \lnot \chi^\tau_q$.
  By our assumption $\frakT \models q$, 
  i.e.~$\frakT \models Q_i$ for some \BCQ $Q_i$ in $q$,
  which means that there is a homomorphism $\eta: \calQ_i \to \frakT$. 
  Given that $\frakT$ is guarded tree decomposable, so is $\eta$. 
  By (i) therefore $\frakT \models \chi^\tau_q$, 
  contradicting our assumption. 
  
{(iii)}
  Recall that, for a \BCQ\ $Q$, the formula $\chi^\tau_Q$ is a disjunction 
  of several (acyclic) \BCQ{} $T$, each of which has at most $3|Q|$ many 
  atoms and therefore requires no more than $3|Q|w$ many variables; 
  the overall number of constituent \BCQ\ of these dimensions is bounded 
  by $(r (3|Q|w)^w)^{3|Q|}$, and each such $T$ has length $\calO(|Q|w)$. 
  All in all, 
  $|\chi^\tau_Q| \leq (r (3|Q|w)^w)^{3|Q|} \calO(|Q|w) = r^{\calO(|Q|)} (|Q|w)^{\calO(|Q|w)}$.
   
  For a \UCQ\ $q=\bigvee_i Q_i$ we have, by definition, 
  $\chi^\tau_q=\bigvee_i \chi^\tau_{Q_i}$, and, if the height of $q$ is $h$, 
  then $|\chi^\tau_q|= r^{\calO(h)} (hw)^{\calO(hw)}$ by the previous estimate. 
  One way to compute $\chi^\tau_q$ is to exhaustively enumerate all acyclic \CQ\ 
  of the right dimensions and to check each one for entailment of some $Q_i$ 
  (verifiable in time $(hw)^{\calO(hw)}$ for each $i$). 
  Such a procedure can be carried out in time $|q| r^{\calO(h)} (hw)^{\calO(hw)}$. 
\end{proof}

Concerning the size of treeifications, note that for a fixed signature 
the figure from~(iii) simplifies to $|\chi^\tau_q| = h^{\calO(h)}$ and 
that a $2^{\Omega(h)}$ lower bound can be established even if we require 
treeifications to be free of redundant disjuncts. 
Indeed, in the signature $\tau = \{E,T\}$, where $E$ is binary and $T$ 
is ternary, it is easy to see that the \BCQ\ $Q_n$ for which $\calQ_n$ is 
a simple $E$-cycle of length $n$, the number of triangulations of $\calQ_n$ 
and hence the number of disjuncts in $\chi^\tau_{Q_n}$ is $2^{\Omega(n)}$. \\

The next key fact is a direct consequence of 
Lemma~\ref{lemma_treeification}~(i) and Definition~\ref{def_cover} 
that highlights the role of query treeification and motivates our 
interest in weakly $N$-acyclic covers.

\begin{fact} \label{fact_covers_and_treeification}
For every weakly $N$-acylic cover $\pi \colon \frakB \covers \frakA$ 
of $\tau$-structures, for every $\varphi \in \GFO[\tau]$ 
and every $q \in \UCQ[\tau]$ of height at most $N$: 
\begin{equation} \label{eq_covers_and_treeification}
  \frakB \models q  \ \Longrightarrow \  \frakA \models \chi^\tau_q
  \qquad \text{ and hence }\qquad 
  \frakA \models \varphi \land \lnot\,\chi^\tau_q  
  \ \Longrightarrow \  
  \frakB \models \varphi \land \lnot\,q
  \ .
\end{equation}
\end{fact}

Using this fact and the finite model property of the guarded fragment, 
Theorem~\ref{thrm_FinControl} will follow, once it is established that 
every finite relational structure admits suitably sized finite 
weakly $N$-acyclic covers for all $N$. That is precisely the content of 
Theorem~\ref{thrm_main}.
Before engaging in the proof of this main technical result let us point out 
a noteworthy consequence of item~(ii) of Lemma~\ref{lemma_treeification}.


\subsubsection*{An interpolation property}

Consider the fragments $\GF$ and $\UCQ$ (equivalently, the positive existential fragment) 
of first-order logic in a relational signature. They are incomparable with respect to 
expressive power and, as mentioned above, the intersection of the two fragments comprises 
(up to semantic equivalence) precisely the unions of guarded conjunctive queries, 
or unions of acyclic conjunctive queries ($\ACQ$). 
In one reading, Lemma~\ref{lemma_treeification}~(ii) states that the fragments $\GF$ and $\UCQ$ 
have a strong form of interpolation with $\ACQ$ interpolants.\footnote{We thank Damian Niwinski 
for this observation.} Indeed, consider some $\varphi \in \GF$ and $q \in \UCQ$ in signature $\tau$. 
Then 
\begin{equation}\label{eq_interpol}
  \varphi \models q   \qquad \Longrightarrow \qquad   \varphi \models \chi^\tau_q \ \text{ and }\  \chi^\tau_q \models q
\end{equation}
and it is interesting to note that the treeification $\chi^\tau_q$ of the query $q$ 
is a uniform interpolant for all $\varphi \in \GFO[\tau]$ that entail $q$.
As a consequence of our Theorem~\ref{thrm_FinControl} we will see that the interpolation 
property~\eqref{eq_interpol} remains intact when the
semantics is restricted to finite models.


\section{The Rosati cover} \label{sec_Rosati}


Rosati proved Proposition~\ref{prop_Rosati} using a ``finite chase''  
procedure~\cite{Ros06,Ros11} that safely reuses variables and results 
in very compact finite models. 
However, his proof of correctness of the finite chase with respect to 
conjunctive query answering is very intricate. 
We adapt the core idea of his model construction to give a more general 
guarded bisimilar cover construction for finite models, and a 
conceptually cleaner and simpler proof of faithfulness with respect to \  
conjunctive queries of bounded size.

\begin{thm} \label{thrm_main_from_gbis_invariant}
Given a bisimulation invariant $\frakI=\Inv(\frakA)$ of an unspecified hypergraph $\frakA$, 
for all $N \geq 2$ one can construct hypergraphs $\frakR_N$ such that $I(\frakR_N) = \frakI$, 
each $\frakR_N$ is $N$-conformal and each $\frakR_{N^2}$ is a weakly $N$-acyclic cover of \ $\frakR_N$. 
Moreover, $|\frakR_{N}| = |\frakI|^{w^{\calO(N)}}$, where $w$ is the 
width~\footnote{Note that this width is determined by $\frakI$.} 
of $\frakA$ and, for fixed $w$ and $N$, $\frakR_N$ 
can be computed in polynomial time from $\frakI$.
The analogous claim for the guarded bisimulation invariant
$\frakI =  \Inv(\frakA)$ of a finite relational structure $\frakA$, 
and concerning guarded bisimilar covers, follows.
\end{thm}

It is not hard to see how this formulation entails the statement 
of Theorem~\ref{thrm_main} as given in the introduction. 
Observe that $\Inv(\frakA)=\Gam(\frakA)$, where $\Gam(\frakA)$ is the  
bisimulation game graph of the given $\frakA$ (without passage to
a non-trivial quotient), can be enforced by introducing 
new predicates to distinguish each individual guarded tuple of $\frakA$. 
Then $\frakA^{(N)}=\frakR_{N^2}$ is a weakly $N$-acyclic (guarded) 
bisimilar cover of $\frakA$ itself and has size $|\frakA|^{w^{\calO(N^2)}}$.
Since $N$-conformality implies conformality at large 
if $N>w=\mathrm{width}(\frakA)$, moreover  
$\frakR_{w+1}$~is a conformal cover of $\frakA$ and we also obtain 
Corollary~\ref{coroll_conformal covers}.

We first define, for every $N$, the $N$-th Rosati cover of a given 
finite hypergraph (or relational structure). After preliminary observations 
much resembling some of Rosati's key lemmas~\cite{Ros06,Ros11} we prove 
their two crucial properties: $N$-conformality of $\frakR_N$ 
and the weak $N$-acyclicity of $\frakR_{N^2}$ as a cover over $\frakR_N$. 
In fact, the structures $\frakR_N$ will form a chain of covers of increasing 
degrees of conformality and weak acyclicity, similar to the construction of~\cite{Otto09rep}.

\subsection{The definition of $\boldmath{\frakR_N}$} \label{sec_RosatiDef}

Let $w$ be the width of $\frakA$ (apparent from $\frakI=\Inv(\frakA)$), 
i.e.~the maximal size of any of its hyperedges (guarded sets). 
We assume throughout that $w > 1$, since width~$1$ is trivial.
For the rest of this section we also fix $m \geq N^2$, $N \geq 2$.

Consider a relational structure $\frakA$ and its guarded-bisimulation invariant $\frakI=\Inv(\frakA)$.
Recall that the vertices of $\Inv(\frakA)$ represent complete $\GF$-types of 
maximal guarded tuples $\tup{a}$ such that $a_i \neq a_j$ for all $i\neq j$ and 
are labeled by the isomorphism type of the sub-structure induced by any (and all) 
corresponding tuple. We denote vertices of $\frakI$ by symbols $d,e,\ldots$ and 
for each $e=[\tup{a}]_{\gbisim}$ we let $[e]=\{1,\ldots,|\tup{a}|\}$. 
Edges of $\frakI$ are triples $\rho=(d,[\rho],e)$, 
where $d=[\tup{a}]_{\gbisim}$ and $e=[\tup{b}]_{\gbisim}$ 
and the label $[\rho]$ is a non-empty partial injection $[d]\to[e]$ such that 
$a_i = b_j$ for all $(i,j)\in[\rho]$. (In this case $\rho$ induces a $\gbisim$-preserving 
partial isomorphism $\tup{a}'\restriction{\dom[\rho]} \to \tup{b}'\restriction{\img[\rho]}$ 
for any $\tup{a}' \gbisim \tup{a}$ and any $\tup{b}' \gbisim \tup{b}$.)
Symbols $\rho,\sigma,$ etc.~will refer to edges of $\frakI$, their respective 
labels will be denoted $[\rho],[\sigma]$, etc. 

We adapt the same notation in the case of a hypergraph and its hypergraph bisimulation 
invariant. Recall that the latter is the bisimulation quotient of the hypergraph 
bisimulation game graph as described earlier. 
%
In the following we often blur the distinction in phrasing and notation between the cases 
of hypergraphs and of relational structures, opting to treat these perfectly analogous cases as one. 

We associate to the invariant $\frakI$ a set of constant and function symbols as follows. 
\begin{list}{\labelitemi}{\leftmargin=1.7em}
\item To every vertex $e$, every $i \in [e]$ and $0 \leq j < w^{m+2}$ 
      we associate a constant symbol $c^j_{e,i}$.
\item To every edge $\rho=(d,[\rho],e)$ 
      and every $i \in [e] \setminus \img[\rho]$ and $0 \leq j < w^{m+2}$ 
      we associate a function symbol $f^j_{\rho,i}$ of arity $|\dom[\rho]|$.
\end{list}
We work with well-formed terms in the above signature. 
As shorthand we write $\mathbf{c}^j_e$ for $(c^j_{e,i})_{1 \leq i\leq k}$, 
and $\mathbf{f}^j_\rho(\tup{t})$ for $(f^j_{\rho,i}(\tup{t}))_{i\not\in\img[\rho]}$ 
and for every tuple $\tup{t} = (t_1,\ldots,t_l)$ we let $\{\tup{t}\}$ stand 
for $\{t_1,\ldots,t_l\}$. For each term $t$ let $J(t)$ denote the set of ``$j$-values'' 
occurring in the superscript of a function symbol at any depth within $t$. 
This notion extends naturally to tuples of terms. Thus $J(c^j_{e,i})=\{j\}$ 
and $J(f^j_{\rho,i}(\tup{t})) = \{j\}\cup J(\tup{t})$. 
The \emph{truncation} of a term $t$ at depth $\kappa$, denoted $t/_\kappa$, 
is defined by the following recursive rules and is extended to tuples of terms 
and to sets of terms in the obvious way.
\begin{equation} \label{eq_trunc}
  {c^j_{e,i}}/_{\kappa} =  c^j_{e,i}  
  \qquad 
  \begin{array}{lcll}
  f^j_{\rho,i}(\tup{t})/_{0} & = & c^j_{e,i} & ( \rho=(d,[\rho],e)) \\
  f^j_{\rho,i}(\tup{t})/_{\kappa+1} & = & f^j_{\rho,i}(\tup{t}/_{\kappa}) 
  \end{array}
\end{equation}
The $N$-th Rosati cover $\frakR_N$ is made up of terms of height at most $N$ 
and is built to realise all guarded bisimulation types in $\frakI$. 
To that end we first define the sets $\calK^r_N(e)$ of ``instances of $e$ at height~$r$'' 
for each $e \in \frakI$ and $r \geq 0$ by simultaneous recursion. 
\begin{equation} \label{eq_KrNe}
\begin{array}{rcll}
  \calK^0_N(e)     &=& \{ \, \mathbf{c}^j_e \, \mid \, j < w^{m+2} \} & \\[.5em]
  \calK^{r+1}_N(e) &=& \{ \, \fatrho^j(\tup{s}\restriction{\dom[\rho]}) \ \, \mid  \, 
                              & \tup{s} \in \calK^r_N(d), \, \rho=(d,[\rho],e), \\ 
                   & &        & j < w^{m+2}, \, j \not\in J(\tup{s}\restriction{\dom[\rho]}) \, \} 
\end{array}
\end{equation}
where for each edge $\rho=(d,[\rho],e)$ in $\frakI$ and terms $\tup{s}$ and $j$ 
as appropriate $\fatrho^j(\tup{s}\restriction{\dom[\rho]})$ denotes 
the tuple $(u_1,\ldots,u_{|[e]|})$ such that 
$$
   u_i = \left\{ \begin{array}{ll} 
                 s_l & (\, (l,i) \in [\rho] \,)                 \\[.5em]
                 f^j_{\rho,i}(\tup{s}_{/N-1}) & (\, i \not\in \img[\rho] \,) 
                 \end{array} \right. 
   \text{ for each } i \in [e] \ .
$$ 
Obviously $\fatrho^j(\tup{s}\restriction{\dom[\rho]})$ depends solely 
on $\tup{t}=\tup{s}\restriction{\dom[\rho]}$, wherefore more often than not 
we shall simply write $\fatrho^j(\tup{t})$ so that, in particular, 
$\{\fatrho^j(\tup{t})\} = \{\mathbf{f}^j_\rho(\tup{t}_{/N-1})\}\cup\{\tup{t}\}$. 

\smallskip

The sets $\calH^r_N(e)$ of \emph{hyperedges above $e$ (at height $r$)} 
are obtained from the above by simply forgetting the tuple ordering:  
  $\calH^r_N(e) = \{ \, \{\tup{t}\} \, \mid \, \tup{t} \in \calK^r_N(e) \, \}$;  
further set $\calH_N(e) = \bigcup_r \calH^r_N(e)$ and $\calH_N = \bigcup_{e\in\frakI} \calH_N(e)$.
All terms $t$ appearing in some hyperedge in $\calH_N$ have height at most $N$ and, 
due to the stipulation $j \not\in J(\tup{t})$ in \eqref{eq_KrNe},  
no function symbol at the root of a subterm of $t$ occurs again within that subterm. 
 
Observe that every $h \in \calH_N$ is either of the form 
  $\{\fatrho^j(\tup{t})\} = \{\mathbf{f}^j_\rho(\tup{t}_{/N-1})\}\cup\{\tup{t}\} \in \calH^{r+1}_N(e)$ 
for some $r$ and $e$ the target of $\rho$ or is equal to some $\{\mathbf{c}^j_e\} \in \calH^0_N(e)$ .
Crucially, under the assumption $N \geq 2$ the constraint $j \not\in J(\tup{t})$ of \eqref{eq_KrNe} 
ensures that the former partitioning of $h$ is unique and we say that $h \in \calH^{r+1}_N(e)$ 
is obtained by \emph{$\rho$-extension} of some (not necessarily unique) 
hyperedge $h' \in \calH^r_N(d)$, with $\rho=(d,[\rho],e)$, 
and denote this using the shorthand $h' \rhoext h$.
Note, in particular, that the sets $\calH(e)$ partition $\calH$.
Henceforth we often omit the subscript $N$ writing $\calH$, $\calH(e)$, etc. 

\smallskip

A hyperedge $h$ will be called a \emph{primary guard of $X$} 
if it is a guard of $X$, viz. $X \subseteq h$, 
and is not the $\rho$-extension of some $h'$ also guarding $X$.

\begin{lem} \label{lemma_primguard}
Assume $m \geq N \geq 2$. 
Then for every guarded set $X$ of terms there is an $e_X \in \frakI$ 
such that all primary guards of $X$ belong to $\calH(e_X)$.
\end{lem}

\begin{proof}
Consider a hyperedge $h$ that is a primary guard of $X$. 
If $h = \{\mathbf{c}^j_e\}$ for appropriate $e$ and $j$, 
then $h$ is the only primary guard of $X$, and we can set $e_X = e$. 
Otherwise we have $h = \{\fatrho^j(\tup{t})\} = 
  \{\mathbf{f}^j_\rho(\tup{t}_{/N-1})\}\cup\{\tup{t}\} \in \calH^{r+1}(e)$
for an appropriate edge $\rho: d \to e$ in $\frakI$, some superscript $j$, 
and terms $\tup{t}$.
Because $h$ is by choice a primary guard of $X$, it cannot be 
that $X \subseteq \{\tup{t}\}$. 
So there is some $f^j_{\rho,i}(\tup{t}/_{N-1})$ in $X$,
and we set $e_X=e$ to be the target of $\rho$.

Suppose indirectly that our choice of $e_X$ was not unique, i.e.~that 
there is some $e' \neq e$ and a primary guard of $X$ of the form
$h' = \{\mathbf{f}^{j'}_{\sigma}(\tup{s}_{/N-1})\}\cup\{\tup{s}\} \in \calH^{r+1}(e')$.  
Then, by the previous argument, some $f^{j'}_{\sigma,i'}(\tup{s}_{/N-1})$ 
would have to be in $X$.
This, however, would imply that $f^j_{\rho,i}(\tup{t}/_{N-2})$ 
had to be among $\tup{s}/_{N-1}$ and vice versa 
$f^{j'}_{\sigma,i'}(\tup{s}_{/N-2})$ among $\tup{t}/_{N-1}$.
Given that $N \geq 2$ this would contradict the requirement 
that $j$ and $j'$ each have but one occurrence in these terms.
\end{proof}

Let $\overline{\calH}_N$ be comprised of the hyperedges in $\calH_N$ together with
sub-hyperedges $h' \subseteq h$  for each $h \in \calH(e)$ precisely as specified 
by the type $\tau_e$ labeling $e \in \frakI$. It follows from the above that whether 
some such $h'$ is included in $\overline{\calH}_N$ does not depend on the choice of $h$.
Indeed, by Lemma~\ref{lemma_primguard}, we may assume that $h$ is a primary guard of $h'$
since for every $\rho=(d,[\rho],e)$ the types ${\tau_d}|_{\dom[\rho]}$ and ${\tau_e}|_{\img[\rho]}$
are identical. \\

\begin{defi}[Rosati cover]   ~\\
We define \,$\frakR^m_N$\, as having universe \,$\bigcup \calH_N$\, 
and hyperedges \,$\overline{\calH}_N$. 
\end{defi}

For the purposes of Theorem~\ref{thrm_main_from_gbis_invariant} 
we shall take $\frakR_N = \frakR^m_N$ and
$\frakR_{N^2}=\frakR^m_{N^2}$ with $m=N^2$.

Using similar reasoning as in Lemma~\ref{lemma_primguard} one can verify 
that $\frakI$ is indeed the guarded bisimulation invariant of $\frakR^m_N$,
i.e., that $\frakR^m_N \gbisim \frakA$ for any $\frakA$ 
with $\Inv(\frakA) = \frakI$.

\begin{lem} \label{lemma_gbisim}
For all $m \geq N \geq 2$ it holds that $\Inv(\frakR^m_N) = \frakI$. 
In particular, for each $e \in \frakI$ all hyperedges in $\calH_N(e)$ 
realise the guarded bisimulation type represented by $e \in \frakI$.
\end{lem}

\begin{proof}
Consider $h_0 \in \calH(e_0)$ and $g_0 \in \calH(d_0)$ such that $X = h_0 \cap g_0 \neq \emptyset$.
As in Lemma~\ref{lemma_primguard} we can find primary guards $h_r, g_s \in \calH(e_X)$ 
of $X$ by tracing backward from $h_0$ and from $g_0$, respectively, 
through extension sequences 
\[\begin{array}{l}
 h_0 \, {\buildrel \rho_1 \over \longleftarrow} \, h_1 \, 
         \, {\buildrel \rho_2 \over \longleftarrow} \, h_2 \, 
         \cdots \, {\buildrel \rho_r \over \longleftarrow} \, h_r \in \calH(e_X) \ \text{ and } \\
 g_0 \, {\buildrel \sigma_1 \over \longleftarrow} \, g_1 \, 
         \, {\buildrel \sigma_2 \over \longleftarrow} \, g_2 \, 
         \cdots \, {\buildrel \sigma_s \over \longleftarrow} \, g_s \in \calH(e_X) \ .
\end{array}\]
Let $h_i \in \calH(e_i)$ for all $0\leq i <r$ and $g_l \in \calH(d_l)$ for all $0 \leq l < s$.
Then in $\frakI$ we have the following paths.
$$
  e_0 \ {\buildrel [\rho_1] \over \longleftarrow} \ e_1 \ \cdots \ 
  e_{r-1} \ {\buildrel [\rho_r] \over \longleftarrow} \ 
  e_X \ {\buildrel [\sigma_s] \over \longrightarrow} \ d_{s-1} \cdots \ 
  d_{1} \ {\buildrel [\sigma_1] \over \longrightarrow} \ d_0
$$
Given the nature of edges in a (guarded) bisimulation invariant as representing 
partial isomorphisms they are invertible and compositional in the sense that 
for each $v \,{\buildrel [\rho] \over \longrightarrow}\, w$ there is 
also $w \, {\buildrel [\rho]^{-1} \over \longrightarrow}\, v$
and then for every $w \,{\buildrel [\sigma] \over \longrightarrow}\, u$ 
there is also $v \, {\buildrel [\sigma]\circ[\rho] \over \longrightarrow}\, u$
as long as $[\sigma]\circ[\rho] \neq \emptyset$. 
This means that for any non-empty $[\pi] \subseteq 
  [\sigma_1] \circ \cdots \circ [\sigma_s] \circ [\rho_r]^{-1} \circ \cdots \circ [\rho_1]^{-1}$
there is an edge $\pi=(e_0,[\pi],d_0)$ in $\frakI$ 
and now there is one such $[\pi]$ that maps the projection of $X$ in $e_0$ 
to the projection of $X$ in $d_0$. 

It follows that all moves made from any $h_0 \in \calH_N(e_0)$ to any $g_0 \in \calH_N(d_0)$
in the guarded bisimulation game on $\frakR^m_N$ have corresponding edges from $e_0$ to $d_0$ in $\frakI$.
The converse of this being enforced by the very definition of $\frakR^m_N$, 
we can establish that the (guarded) bisimulation invariant of $\frakR^m_N$
is no other than $\frakI$.
\end{proof}

\begin{lem} \label{lemma_chain}
$\calH^r_N(e)/_{k} = \calH^r_{k}(e)$ for all $m \geq N > k \geq 2$, 
all $r$, and all $e \in \frakI$. Truncation of terms at depth $k$ thus acts   
as a homomorphic projection from $\frakR^m_N$ onto $\frakR^m_k$ inducing 
a guarded bisimulation. Therefore, for every $N\leq m$,  
we have the following chain of covers.
$$
   \frakR^m_N \,\longcovers\, \frakR^m_{N-1} \, \longcovers\, \cdots \,\frakR^m_3\, \longcovers \, \frakR^m_2  
$$
\end{lem}

\begin{proof}
For $\sigma \in \mathrm{Sym}([w^{m+2}])$ a permutation of $j$-values 
and $t$ a term let $t^\sigma$ denote the term obtained by translating 
all superscripts $j$ in $t$ according to $\sigma$.
\[\begin{array}{rcl}
 (c^j_{e,i})^\sigma  & = & c^{\sigma(j)}_{e,i} \\
 (f^j_{\rho,i}(\tup{t}))^\sigma  & = &  f^{\sigma(j)}_{\rho,i}(\tup{t}^\sigma)
\end{array}\]
Based on definitions \eqref{eq_trunc} and \eqref{eq_KrNe} it is straightforward to verify 
by induction on $N$ and on $r$ that $\calH^r_N(e)/_{N-1} \subseteq \calH^r_{N-1}(e)$ 
and that $\calH^r_N(e)$ is closed under translations $\cdot^\sigma$ for all $e$.

Using the latter one can in fact show by induction that 
$\calH^r_N(e)/_{N-1} = \calH^r_{N-1}(e)$ for all $r$, $e \in \frakI$ and $m\geq N$. 
This amounts to proving that all hyperedges 
  $h = \{\fatrho^j(\tup{t}|_{\dom[\rho]})\} \in \calH^{r+1}_{N-1}(e)$ 
obtained by $\rho$-extension of some $g = \{\tup{t}\} \in \calH^r_{N-1}(d)$ 
can also be obtained as truncations of hyperedges in $\calH^{r+1}_N(e)$, assuming, 
by the induction hypothesis, that $\calH^r_{N-1}(d) = \calH^r_N(d)/_{N-1}$, 
i.e., that there is a $\hat{g} = \{\tup{u}\} \in \calH^r_N(d)$ 
such that $g = \hat{g}/_{N-1}$.
While $j \not\in J(\tup{t}|_{\dom[\rho]})$ in this case, cf. \eqref{eq_KrNe}, 
it is conceivable that $j$ does occur in $\tup{u}|_{\dom[\rho]}$ at depth $N$.  
If so, then take a permutation $\sigma \in \mathrm{Sym}([w^{m+2}])$ 
that fixes $J(\tup{u}|_{\dom[\rho]})\setminus\{j\}$ pointwise but does not fix $j$
(that such a permutation exists follows from $N \leq m$), 
otherwise let $\sigma = \mathrm{id}$. 
Then, by closure under translations, ${\hat{g}}^\sigma = \{\tup{u}^\sigma\} \in \calH^r_N(d)$ 
and by the choice of $\sigma$ we have $\tup{u}^\sigma/_{N-1} = \tup{u}/_{N-1}$ 
and $j \not\in J(\tup{u}^\sigma|_{\dom[\rho]})$.
Consequently $\hat{h} = \{\fatrho^j(\tup{u}^\sigma|_{\dom[\rho]})\}$ is 
a hyperedge in $\calH^{r+1}_N(e)$ and  
$\hat{h}/_{N-1} = 
       \{\mathbf{f}_\rho^j(\tup{u}^\sigma|_{\dom[\rho]}/_{N-2})\}
  \cup \{\tup{u}^\sigma|_{\dom[\rho]}/_{N-1}\} = h$ as needed.

It follows that $\calH^r_N(e)/_{\kappa} = \calH^r_{\kappa}(e)$ for all $m \geq N > \kappa$. 
Truncation of terms at depth $N-1$ is therefore a homomorphism 
from $\frakR^m_N$ to $\frakR^m_{N-1}$ that is onto. 
By Lemma~\ref{lemma_gbisim} it also induces a (guarded) bisimulation 
$\frakR^m_N \longcovers \frakR^m_{N-1}$ yielding a chain of covers as claimed. 
\end{proof}

\subsection{Size of the Rosati cover}

The size of $\frakR^m_N$ can be bounded as follows. 
Let $w$ be the width of $\frakI$, assume that $w \geq 2$ and let $J=w^{m+2}$.
Then there are $J |\frakA|^{\calO(w)}$ many constants $c^j_{e,i}$ 
and function symbols $f^j_{\rho,i}$ altogether, and each term of height up to $N$ 
contains at most $w^{N+1}$ many such symbols. 
For $m=N^2$, therefore, the total number of terms in $\frakR^m_{N}$ 
is at most $(J |\frakI|^{\calO(w)})^{w^{N+1}} = |\frakI|^{w^{\calO(N)}}$ 
as stated in Theorem~\ref{thrm_main_from_gbis_invariant}.

\subsection{Auxiliary notions}

Consider a hyperedge $h=\{\mathbf{f}^j_\rho(\tup{t}_{/N-1})\} 
\cup \{ \tup{t}\} \in \calH^{r+1}_N(e)$. 
The elements of $\{\mathbf{f}^j_\rho(\tup{t}_{/N-1})\}$ 
will be referred to as \emph{siblings}; we denote the sibling relation 
as $f^j_{\rho,i}(\tup{t}_{/N-1}) \sib f^j_{\rho,l}(\tup{t}_{/N-1})$. 
We also say that these terms are \emph{introduced in the hyperedge $h$} 
and that \emph{$h$ is a $\rho$-extension}. 
Furthermore, elements of $\{ \tup{t} \}$ are said to be \emph{predecessors} 
of those in $\{\mathbf{f}^j_\rho(\tup{t}_{/N-1})\}$, and we denote this 
by writing $t_l \pred f^j_{e,i}(\tup{t}_{/N-1})$, for $l$ and $i$ as appropriate.
Constants covered by a hyperedge $\{\mathbf{c}^j_e\} \in \calH^0_N(e)$ 
are also regarded as siblings introduced in that hyperedge. 
%
Compare Lemmas~4--9 of~\cite{Ros11} for some of the following properties. 

\begin{lem} \label{lemma_props} 
Let $m \geq N \geq 2$ as before.
\begin{enumerate}[label=(\roman*)]
\setlength{\itemsep}{0.2ex}
\item \label{props_partition}
      The relations $\sib$, $\pred$, and its inverse $\predinv$ 
      partition the set of all guarded pairs of $\frakR^m_N$.
\item \label{props_sib_equiv}
      $\sib$ is an equivalence relation having guarded equivalence classes. 
\item \label{props_pred_sib} 
      Whenever $t^0 \pred t^1 \sib t^2$ then $\{t^0,t^1,t^2\}$ is guarded and $t^0 \pred t^2$.
\item \label{props_pred_Nacylic} 
      $\frakR^m_N$ has no directed $\pred$-cycles of length $\leq N$.
\item \label{props_prim_max}
      If $h$ is a primary guard of $X$ 
      then some $\pred$-maximal element of $h$ must be in $X$. 
\item \label{props_pred_trans} 
      Assuming $m \geq N \geq 3$, the relation $\pred$ 
      is transitive on every guarded set of terms. 
\item \label{props_prim_trunc} 
      If $m \geq N \geq 3$ and  
      $h \in \frakR^m_N$ is a (primary) guard of $X \subseteq \frakR^m_N$
      then $h/_{N-1}$ is a (primary) guard of $X/_{N-1} \subseteq \frakR^m_{N-1}$.
      In particular, $e_X = e_{X/_{N-1}}$ for every guarded set $X \subseteq \frakR^m_N$. 
\end{enumerate}
\end{lem}
\begin{proof}
As to item \eqref{props_partition}, observe that the sibling and predecessor relationships are 
reflected in the terms themselves. 
Siblings are identical terms for all but the indices in the subscript of their 
respective root symbol, and the $(N-1)$-truncation of each predecessor 
of a term occurs in it as an immediate subterm of the root symbol. 
Given that $N \geq 2$ and that the $j$ superscripts are by definition unique 
within each term, it is impossible for some $t \in \frakR^m_N$ to have $t/_{N-2}$ 
as a subterm at depth $2$, and hence it is impossible to have some $t \pred t' \pred t$.  

Item \eqref{props_sib_equiv} can be equivalently stated in a form similar to 
that of \eqref{props_pred_sib}, asserting that whenever $t^0 \sib t^1 \sib t^2$ 
then  $\{t^0,t^1,t^2\}$ is guarded and also $t^0 \sib t^2$ holds. 
Let us first verify that $\{ t^0, t^1, t^2 \}$ is guarded in both these cases.
For item \eqref{props_sib_equiv} this is obviously the case if $t^0,t^1,t^2$ 
are sibling constants belonging to some $\mathbf{c}^j_e$. Otherwise let 
$h = \{\mathbf{f}^j_\rho(\tup{u}_{/N-1})\}\cup\{\tup{u}\} \in \calH^{r+1}(e_{\{t^0,t^1\}})$
be any primary guard of the pair $\{t^0,t^1\}$. 
Then, whether $t^0 \sib t^1$ \eqref{props_sib_equiv} or $t^0 \pred t^1$ \eqref{props_pred_sib} 
the term $t^1$ must have been introduced in the hyperedge $h$ and must therefore 
take the form $t^1 = f^j_{\rho,i}(\tup{u}/_{N-1})$.
Being a sibling of $t^1$,  $t^2 = f^j_{\rho,l}(\tup{u}/_{N-1})$ and as such is contained in $h$, 
which therefore guards $\{ t^0, t^1, t^2 \}$. Now it is obvious from the definition of 
the sibling and predecessor relations that $t^0 \sib t^2$ or $t^0 \pred t^2$, 
according to whether $t^0 \sib t^1$ or $t^0 \pred t^1$.

Item \eqref{props_pred_Nacylic} is a trivial consequence of the requirement that 
superscripts $j$ must not occur twice in any term in $\frakR^m_N$.  
Indeed, if $t^{k-1} \pred \ldots \pred t^2 \pred t^1 \pred t^0$ is a 
predecessor chain of length $k \leq N$ then $t^r/_{N-r}$ is, for each $r<k$, 
a subterm of $t^0$ at depth $r$. 
This implies that $t^0, t^1, \ldots, t^{k-1}$ are pairwise distinct. 

Property \eqref{props_prim_max} is a straightforward consequence of the definitions. 
Consider $h$ a primary guard of $X$. 
Either $h = \{\mathbf{c}^j_e\}$ and each $c^j_{e,i}$ is $\pred$-maximal within $h$, 
or $h=\{\mathbf{f}^j_\rho(\tup{t}/_{N-1})\}\cup\{\tup{t}\}$ introduces some 
$f^j_{\rho,i}(\tup{t}/_{N-1}) \in X$, which is then $\pred$-maximal within $h$.

Assuming $m \geq N \geq 3$, property \eqref{props_pred_trans} follows 
from the prior ones. Indeed, let $t^0 \pred t^1 \pred t^2$ such that 
$\{t^0, t^1, t^2\}$ is guarded. Then, according to \eqref{props_partition} 
either $t^2 \sib t^0$ or $t^2 \pred t^0$ or $t^0 \pred t^2$.
In the first case we have $t^1 \pred t^0$ by \eqref{props_pred_sib}
and thus a two-cycle $t^0 \pred t^1 \pred t^0$, 
in the second case we have a three-cycle $t^0 \pred t^1 \pred t^2 \pred t^0$, 
both contradicting \eqref{props_pred_Nacylic}. 
Therefore $t^0 \pred t^2$. 

Finally, towards \eqref{props_prim_trunc} consider a hyperedge $h$ 
that is a guard (i.e.~superset) of $X \subseteq \frakR^m_N$.
By Lemma~\ref{lemma_chain}, also $h/_{N-1}$ is a hyperedge in $\frakR^m_{N-1}$,
and it guards $X/_{N-1}$. 
If $h/_{N-1}$ is \emph{not} a primary guard of $X/_{N-1}$ 
then $h/_{N-1}$ is of the form 
$\{\mathbf{f}^j_\rho(\tup{t}/_{N-2})\}\cup\{\tup{t}\} \in \calH^{r+1}_{N-1}(e)$
and $X/_{N-1} \subseteq \{\tup{t}\}$. 
By Lemma~\ref{lemma_chain} again, 
$h = \{\mathbf{f}^j_\rho(\tup{u}/_{N-1})\}\cup\{\tup{u}\} \in \calH^{r+1}_{N}(e)$ 
for some terms $\tup{u}$ such that $\tup{u}/_{N-1} = \tup{t}$.
Suppose now that $h$ is a primary guard of $X$ and thus there is some term
$f^j_{\rho,l}(\tup{u}/_{N-1})$ belonging to $X$. 
Then $f^j_{\rho,l}(\tup{u}/_{N-1})/_{N-1} = f^j_{\rho,l}(\tup{u/}_{N-2}) = f^j_{\rho,l}(\tup{t}/_{N-2})$
belongs to $X/_{N-1}$. 
Given that $N-1 \geq 2$, this contradicts the assumption $X/_{N-1} \subseteq \{\tup{t}\}$,
i.e.~that $h/_{N-1}$ is not a primary guard of $X/_{N-1}$. 
\end{proof}

\subsection{{\boldmath $N$}-conformality of {\boldmath $\frakR_N$}}

Consider $3 \leq l \leq N$ and an $l$-clique $\{t^0, \ldots, t^{l-1}\}$ 
in $\frakR_N$, i.e., such that all pairs $\{t^i,t^j\}$ are guarded.
By Lemma~\ref{lemma_props} there are no predecessor-cycles 
in $\{t^0, \ldots, t^{l-1}\}$ but there is a term, wlog.~$t^0$, 
such that every one of $t^1, \ldots, t^{l-1}$ is either a predecessor 
or a sibling of $t^0$. 

Observe that the projection of any primary guard of $t^0$ to $\frakR^m_{N-1}$ 
guards $\{t^0_{/N-1}, \ldots, t^{l-1}_{/N-1}\}$.
This would already be sufficient to establish a weaker form of Theorem~\ref{thrm_main_from_gbis_invariant}
still yielding Theorem~\ref{thrm_FinControl} for $\GFO$.
However, we can show that the entire $l$-clique is guarded already in $\frakR_N$.

\begin{prop} \label{prop_conformality}
Assume that for some $2 \leq l \leq N$ there are $t^0, \ldots, t^{l-1}$ in $\frakR_N$ 
such that all pairs $\{t^i,t^j\}$ are guarded.
Then the entire clique $\{t^0, \ldots, t^{l-1}\}$ is guarded in $\frakR_N$. 
\end{prop}

\begin{proof}
From the trivial base case for $l=2$ we proceed by induction on $l$.
By the preceding observation we may assume wlog.~that each of $t^1, \ldots, t^{l-1}$ 
is either a predecessor or a sibling of $t^0$. 
By the induction hypothesis $X=\{t^1,\ldots,t^{l-1}\}$ is guarded.

Consider first the case when $t^0 \sib t^i$ for some $i \not= 0$. 
Then $t^i$ is a $\pred$-maximal element of $X$ and as such is necessarily 
introduced in any primary guard $h$ of $X$. But then $t^0$, being a sibling of $t^i$, 
is also introduced in $h$, which therefore guards the entire clique.

Otherwise we know that $t^i \pred t^0$ for all $0<i<l$. 
Also, $X$ being guarded it contains a $\pred$-maximal element, wlog. $t^1$. 
Let $h^{(0)}$ be a primary guard of the pair $\{t^0, t^1\}$. 
Given that $t^1 \pred t^0$, then $t^0$ is introduced in $h^{(0)}$, 
cf. property \eqref{props_prim_max}.
In this case $t^0$ takes the form $f^{j_0}_{\rho_0,i_0}(\tup{u}/_{N-1})$ 
for some $\rho_0: e_1 \to e_0$ and appropriate $i_0$ and 
$
  h^{(0)} = \{\fatrho_0^{j_0}(\tup{u})\} =
     \{ \mathbf{f}^{j_0}_{\rho_0}(\tup{u}/_{N-1}) \}
     \cup 
     \{\tup{u}\} \in \calH(e_0)
$
where $\{\tup{u}\} = h^{(1)}|_{\dom\rho_0}$ 
for some $h^{(1)} \in \calH(e_1)$.

Note that each $t^i/_{N-1}$ is a subterm of $t^0$ at depth one, i.e., 
is among those in $\tup{u}/_{N-1}$. 
For each $i$, let $u^i$ denote the component of $\tup{u}$ 
such that $u^i/_{N-1} = t^i/_{N-1}$ and let $Y = \{u^1,\ldots,u^{l-1}\}$. 
Thus $X/_{N-1} = Y/_{N-1}$. 

Crucially $t^1 = u^1$, for it is included in $h^{(1)}$. 
Also observe that, because $t^1$ is $\pred$-maximal in $X$,  
for each $1<i<l$ either $t^i \sib t^1$ or $t^i/_{N-1} = u^i/_{N-1}$ 
also occurs as a subterm of $t^1$.

Tracing backward from $h^{(0)}$ and $h^{(1)}$ as above 
we can find a chain of expansions  
\begin{equation} \label{eq_deriv}
  h^{(0)} \,{\buildrel \rho_0^{j_0} \over \longleftarrow} \, h^{(1)} \, 
            {\buildrel \rho_1^{j_1} \over \longleftarrow} \, \cdots \,
            {\buildrel \rho_r^{j_r} \over \longleftarrow} \, h^{(r+1)}
\end{equation}
with $h^{(\lambda)}$ 
a guard of $Y$ for each $\lambda \leq r$ 
until we reach some $h^{(r+1)} \in \calH(e_Y)$ a primary guard of $Y$.
Let $\{\tup{v}\} = h^{(r+1)}$ and consider some $\{\tup{w}\} = g^{(r+1)} \in \calH(e_X)$, 
a primary guard of $X$.

Given that $N \geq l \geq 3$ and $X/_{N-1} = Y/_{N-1}$, 
according to Lemma~\ref{lemma_props}~\eqref{props_prim_trunc}   
we have $e_X = e_{X/_{N-1}} = e_{Y/_{N-1}} = e_{Y}$, 
thereby $g^{(r+1)}, h^{(r+1)} \in \calH_N(e_X)$. 
From this it follows that the extension sequence 
$
  \langle \rho_r, \ldots, \rho_1, \rho_0 \rangle
$
is applicable to $g^{(r+1)}$. However, our aim is to mimic the exact same 
derivation sequence with the same $j_\lambda$-values as in \eqref{eq_deriv} 
starting from $g^{(r+1)}$.

Notice that $t^1$ being $\pred$-maximal among $X$, it is also $\pred$-maximal in $g^{(r+1)}$,
in accordance with Lemma~\ref{lemma_props}\eqref{props_prim_max}.
Therefore, every $w_k \in g^{(r+1)}$ is either a sibling or a predecessor of $t^1$.
Similarly, every $v_k \in h^{(r+1)}$ is a sibling or a predecessor, respectively, of $u^1$. 
In other words, the exact relationships within $g^{(r+1)}$ are mirrored in $h^{(r+1)}$. 
Given that $t^1=u^1$ this implies $h^{(r+1)}/_{N-1} = g^{(r+1)}/_{N-1}$.

Having ascertained $h^{(r+1)}/_{N-1} = g^{(r+1)}/_{N-1}$, 
it now follows that the same extension sequence 
$
  \langle \rho_r^{j_r}, \ldots, \rho_1^{j_1}, \rho_0^{j_0} \rangle 
$
as in \eqref{eq_deriv} is applicable to $g^{(r+1)}$ -- with the very same $j_\lambda$-values -- 
producing an analogous derivation to that of $h^{(0)}$ from $h^{(r+1)}$: 
$$
  g^{(0)} \,{\buildrel \rho_0^{j_0} \over \longleftarrow} \, g^{(1)} \, 
            {\buildrel \rho_1^{j_1} \over \longleftarrow} \, \cdots \,
            {\buildrel \rho_r^{j_r} \over \longleftarrow} \, g^{(r+1)}
$$
ending in some $g^{(0)} \in \calH(e_0)$.
Note that, because each $h^{(\lambda)}$ is a guard of $Y$, 
also each $g^{(\lambda)}$ is a guard of $X$. Moreover, a simple induction shows 
that $h^{(\lambda)}/_{N-1} = g^{(\lambda)}/_{N-1}$ for all $\lambda \leq r+1$. 
In particular, 
$
  g^{(0)} = \{\fatrho^{j_0}_0(\tup{v})\} = \{\mathbf{f}^{j_0}_{\rho_o}(\tup{v}|_{N-1})\}\cup\{\tup{v}\}
$ 
where $\{\tup{v}\} = g^{(1)}|_{\dom\rho_0}$ and $\tup{v}/_{N-1} = \tup{u}|_{N-1}$.
As such, $g^{(0)}$ introduces
$$ 
  f^{j_0}_{\rho_0,i_0}(\tup{v}/_{N-1}) = 
  f^{j_0}_{\rho_0,i_0}(\tup{u}/_{N-1}) = t^0
$$
and thus guards the entire clique. 
\end{proof}

\subsection{Weak {\boldmath $N$}-acyclicity of {\boldmath $\frakR_{N^2}$} over {\boldmath $\frakR_N$}}

\begin{prop} \label{prop_acyclicity}
Let $M=(N^2+N)/2 \leq m$ and consider a homomorphism $h\colon\frakS\to\frakR^m_{M}$ 
with $|\frakS|\leq N$ and the projection $\pi\colon\frakR^m_{M}\to\frakR^m_N$. 
Then there is a guarded tree-decomposable (not necessarily induced) 
sub-structure $\frakT$ of $\frakR^m_N$ containing the $\pi \circ h$-image of $\frakS$.
\end{prop}

Weak $N$-acyclicity of $\pi\colon\frakR^m_{M}\longcovers\frakR^m_N$ follows since $\pi \circ h$ 
trivially factors through $f=\pi\circ h$ with $f\colon\frakS\to\frakT \subseteq \frakR^m_N$. 
Then obviously also $\frakR^m_{N^2}\longcovers\frakR^m_N$ is weakly $N$-acyclic. 
Following the analogous result in~\cite[Section 5]{Otto09rep}, the proof of the proposition 
is based on the $N$-conformality of $\frakR^m_N$ (cf.~ Proposition~\ref{prop_conformality}) 
and an inductive use of the next observation.

\begin{lem} \label{lemma_chords}
Consider an $l$-cycle 
 $C=\{\{t^0,t^1\},\{t^1,t^2\},\ldots,\{t^{i},t^{i+1}\},\ldots,\{t^{l-1},t^0\}\}$ 
in the Gaifman graph of $\frakR^m_k$, with $3 \leq l \leq k \leq m$.
Then there is some $i$ such that $\{t^{i-1},t^i,t^{i+1}\}/_{k-1}$ is a guarded triangle 
in $\frakR^m_{k-1}$ and $\{t^0,\ldots,t^{i-1}, t^{i+1}, \ldots, t^{l-1}\}_{/k-1}$ 
a cycle of length $l-1$.
\end{lem}

\begin{proof}
Given an $l$-cycle in $\frakR^m_N$ as above,
by Lemma~\ref{lemma_props}~(i) we know that for every $i$ 
either $t^i \sib t^{i+1}$ or $t^i \pred t^{i+1}$ or $t^{i+1} \pred t^i$.
According to Lemma~\ref{lemma_props}~(iv) 
there are no predecessor cycles of length $\leq N$ in $\frakR^m_N$, 
hence it cannot be the case that $t^i \pred t^{i+1}$ for all $i$, 
nor that $t^{i+1} \pred t^{i}$ for all $i$. 
Then for some $i$ one of the following cases must hold:
\begin{itemize}
\setlength{\itemsep}{-0.2ex}
\item $t^{i-1} \sib t^i \sib t^{i+1}$: 
   then, by Lemma~\ref{lemma_props}~\eqref{props_sib_equiv}, 
   $\{t^{i-1}, t^i, t^{i+1}\}$ is guarded and $t^{i-1} \sib t^{i+1}$;
\item $t^{i-1} \pred t^i \sib t^{i+1}$: 
   then, by Lemma~\ref{lemma_props}~\eqref{props_pred_sib}, 
   $\{t^{i-1}, t^i, t^{i+1}\}$ is guarded and $t^{i-1} \pred t^{i+1}$;
\item $t^{i+1} \pred t^i \sib t^{i-1}$: 
   then, similarly, 
   $\{t^{i-1}, t^i, t^{i+1}\}$ is guarded and $t^{i+1} \pred t^{i-1}$;
\item $t^{i-1} \pred t^i \predinv t^{i+1}$: 
   then both $t^{i-1}_{/N-1}$~and~$t^{i+1}_{/N-1}$ are maximal proper subterms of $t^i$; 
   therefore, the projection $h/_{N-1}$ of any hyperedge $h$ of $\frakR^m_N$ 
   in which $t^i$ was introduced guards $\{t^{i-1},t^i,t^{i+1}\}/_{N-1}$ 
   in $\frakR^m_{N-1}$.
\end{itemize}
In each case we will have found some $i$ such that 
$\{t^0,\ldots,t^{i-1}, t^{i+1}, \ldots, t^{l-1}\}_{/N-1}$ 
constitutes in $\frakR^m_{N-1}$ a cycle of length $l-1$ 
and $\{t^{i-1},t^i,t^{i+1}\}/_{N-1}$ a guarded triangle.
\end{proof}

\vskip .5em


\begin{proof}[Proof of Proposition~\ref{prop_acyclicity}]
Recall from Lemma~\ref{lemma_chain} that for $i<M-2$ the structures $\frakR^m_{M-i}\longcovers\frakR^m_{M-i-1}$ 
form a chain of covers induced by the projections $\pi^i\colon(t \to t/_{M-i-1})$. 

Let $\frakS^0=\frakS$ and $h^0=h$. 
For every $0 < i \leq N(N-1)/2$ we shall inductively construct 
a finite structure $\frakS^{i+1} \supseteq \pi^i(h^i(\frakS^i))$ 
and a homomorphism $h^{i+1}\colon\frakS^{i+1}\to\frakR^m_{M-(i+1)}$ 
such that i) the $\pi^i\circ h^i$-image of every chordless cycle 
in the Gaifman graph of $\frakS^i$ contains a chord in $\frakS^{i+1}$ 
and ii) no element of $\frakS^{i+1}\setminus \pi^i(h^i(\frakS^i))$ 
lies on a chordless cycle in the Gaifman graph of $\frakS^{i+1}$. 

If there are no chordless cycles in the Gaifman graph of $\frakS^i$ then 
let $\frakS^{i+1}=\frakS^i$ and let $h^{i+1}$ be the identity on $\frakS^i$.
Note that once this happens, the construction stabilizes in the sense 
that $\frakS^{j}=\frakS^i$ will hold for all $i \leq j \leq N(N-1)/2$.  
Otherwise, from Lemma~\ref{lemma_chords} we know that for every chordless cycle $\frakC$
in the Gaifman graph of $\frakS^i$, its image $\pi^i(h^i(\frakC))$ does have a chord 
in $\frakR^m_{M-i-1}$ in the form of a guarded pair $\{\pi^i(h^i(a)),\pi^i(h^i(b))\}$ 
with $a$,$b$ non-neighboring elements of $\frakC$. 
In the degenerate case $\pi^i(h^i(a))=\pi^i(h^i(b))$; otherwise the chord results 
from some relational atom $\alpha^\frakC=P^\frakC(\bar{a}^\frakC)$ of $\frakR^m_{M-i}$ 
with both $\pi^i(h^i(a))$ and $\pi^i(h^i(b))$ among $\bar{a}^\frakC$. 
To construct $\frakS^{i+1}$ we take $\pi^i(h^i(\frakS^i))$ amalgamated with 
a fresh copy of $\alpha^\frakC$ for every chordless cycle in the Gaifman graph 
of $\frakS^i$, whereby the copy of each $\alpha^\frakC$ is attached 
to $\pi^i(h^i(\frakS^i))$ only at those two elements $\pi^i(h^i(a))$,$\pi^i(h^i(b))$
forming the new-found chord of $\frakC$. This ensures, as asserted, that no element 
of $\frakS^{i+1}$ thus introduced outside of $\pi^i(h^i(\frakS^i))$ lies on a 
chordless cycle in the Gaifman graph of $\frakS^{i+1}$. 
We define $h^{i+1}$ to act as the identity on $\pi^i(h^i(\frakS^i))$ 
and to map the copy of each $\alpha^\frakC$ to the original $\alpha^\frakC$ 
inside $\frakR^m_{M-i-1}$. 

Notice that since every cycle in whatever $\frakS^i$ involves only images of elements 
of $\frakS=\frakS^0$, of which there are no more than $N$, the above process must 
stabilize after at most $L=N(N-1)/2$ many steps, the number of pairs of elements of 
$\frakS$ whose images can potentially arise as chords of any cycle in any $\frakS^i$.

Thus we will have found an $h^{L}\colon\frakS^{L}\to\frakR^m_N$ such that the 
Gaifman graph of $\frakS^{L}$ contains no chordless cycles, viz.\
it is chordal, whence the same can be said of its image in $\frakR^m_N$. 
Given that, by Proposition~\ref{prop_conformality}, $\frakR^m_N$ is $N$-conformal,
it follows that the $h^{L}$-image of $\frakS^{L}$ in $\frakR^m_N$ is contained 
in an acyclic (not-necessarily induced) sub-structure $\frakT$ of $\frakR^m_N$. 
\end{proof}


\section{Finite controllability and small models} \label{sec_FMP}


Relying on Theorem~\ref{thrm_main} one can show that $\UCQ$ answering 
against \GF and even against \CGF sentences is finitely controllable. 
The sharper Theorem~\ref{thrm_main_from_gbis_invariant} also yields optimal 
upper bounds on the minimal size of finite models for each of these fragments 
as expressed in Theorems~\ref{thrm_FinControl}~\&~\ref{thrm_small_models} below. 
Matching lower bounds are implicit in \cite{Gr99JSL}. 


\newtheorem*{theorem_fincntrl}{Theorem~\ref{thrm_FinControl}}
\begin{theorem_fincntrl}
For every $\varphi \in \GFO$ and every $q \in \UCQ$:
\[  
\varphi \models q  \; \iff \;   \varphi \models_{\mathrm{fin}} q. 
\] 
More specifically, if $\varphi \land \lnot\, q$ is satisfiable 
then it has a finite model of size $2^{(|\varphi|+|\tau|^{\calO(h)}) (wh)^{\calO(wh^2)}}$,  
where $h$ is the height of $q$, $\tau$ is the signature of $\varphi$, 
and $w$ the width of $\tau$.
\end{theorem_fincntrl}

\begin{proof}
Recall the properties of $\chi_q$ from Lemma~\ref{lemma_treeification}.
We establish the claim by proving the following equivalences.
$$ 
  \varphi \models q \ \text{ iff }\ \varphi \models \chi_q \ \text{ iff }\  
  \varphi \models_\mathrm{fin} \chi_q \ \text{ iff }\ \varphi \models_\mathrm{fin} q 
$$
The first equivalence was proved in Lemma~\ref{lemma_treeification}~{(ii)}  
and the second equivalence follows from the finite model property 
of the guarded fragment. 
Also $\varphi \models_\mathrm{fin} \chi_q \ \Rightarrow \ \varphi \models_\mathrm{fin} q$
is a trivial consequence of $\chi_q \models q$. 
It remains to be seen that $\varphi \not\models_\mathrm{fin} \chi_q$
implies $\varphi \not\models_\mathrm{fin} q$.
Note that $\varphi \not\models_\mathrm{fin} \chi_q$ is the same as 
$\varphi \not\models \chi_q$ thanks to the finite model property of $\GF$. 

So assume that $\varphi \land \lnot \chi_q$ is satisfiable. 
Then, by Proposition~\ref{prop_satcrit}, there is some invariant $\frakI$ 
satisfying the Scott normal form $\psi$ of $\varphi$ as in Lemma~\ref{lemma_normalform}. 
Let $h$ be the height of $q$, viz.~the maximal size of its consituent \CQ.
Applying Theorem~\ref{thrm_main_from_gbis_invariant} on input $\frakI$ with $N=h$ 
we obtain finite models $\frakR^{N^2}_N$ and $\frakR^{N^2}_{N^2}$ 
of $\varphi \land \lnot \chi_q$, with $\frakR^{N^2}_{N^2}$ 
a weakly $N$-acyclic cover of $\frakR^{N^2}_N$.
From Fact~\ref{fact_covers_and_treeification} it then follows 
that $\frakR^{N^2}_{N^2} \models \varphi \land \lnot q$. 
This concludes the proof of finite controllability. 

According to Theorem~\ref{thrm_main_from_gbis_invariant}, 
  $|\frakR^{h^2}_{h^2}| = |\frakI|^{w^{\calO(h^2)}}$,
where $w$ is the width of the signature $\tau$. 
From Proposition~\ref{prop_satcrit} it follows that $|\frakI|$ 
is bounded by the number of atomic types in the signature of $\psi$,
which is of the order $2^{\calO((|\tau| + |\varphi| + |\chi_q|)w^w)}$.
Finally, Lemma~\ref{lemma_treeification}~(iii) gives 
$|\chi^\tau_q| = |\tau|^{\calO(h)} (hw)^{\calO(hw)}$.
Putting it all together we obtain the estimate 
$|\frakR^{h^2}_{h^2}| 
  = 2^{(|\tau| + |\varphi| + |\chi_q|)w^{\calO(w+h^2)}}
  = 2^{|\varphi|w^{\calO(w+h^2)}+|\tau|^{\calO(h)} (wh)^{\calO(wh^2)}}
  = 2^{(|\varphi|+|\tau|^{\calO(h)}) (wh)^{\calO(wh^2)}}$ as claimed.
\end{proof}

Naturally, both the size and the width of the signature of a \GF-formula 
are bounded by its length. It is thus easy to see how the above statement 
of Theorem~\ref{thrm_FinControl} implies that given in the introduction. 
In particular, we observe the following corollaries.

\begin{cor}
For every $k$, every satisfiable sentence of the $k$-variable guarded fragment 
has finite models of exponential size in the length of the formula.
\end{cor}

\begin{cor}
For a finite set $F$ of $\tau$-structures let $\calC_F$ denote the class 
of those $\tau$-structures not allowing a homomorphic image of any 
member of $F$.
If a guarded sentence $\varphi$ has a model in $\calC_F$ then it also 
has one of size $2^{\calO(|\varphi|)}$.
\end{cor}

Another corollary is the validity of the uniform interpolation 
property \eqref{eq_interpol}
for \GFO and the positive existential fragment also in the finite model semantics. 

\begin{cor}
Consider some $\varphi \in \GF$ and $q \in \UCQ$ in signature $\tau$. Then 
\begin{equation}  
  \varphi \,\models_{\mathrm{fin}}\, q 
  \qquad \Longrightarrow \qquad  
  \varphi \,\models_{\mathrm{fin}}\, \chi^\tau_q \ \ \text{ and } \ \ \chi^\tau_q \,\models_{\mathrm{fin}}\, q
\end{equation}
\end{cor}


\subsection{Finite controllability for the clique-guarded fragment}


The above results easily carry over to the clique-guarded fragment with essentially the same bounds. 
In \cite[Section 3.3]{HO03} a reduction of the (finite) satisfiability problem for $\CGF$ 
to the (finite) satisfiability problem for $\GF$ is presented. 
We borrow their idea with some adaptations to keep the blow-up in formula size to a minimum.

Our reduction maps a given clique-guarded sentence $\varphi \in \CGF[\tau]$ 
to a guarded sentence $\varphi^\ast \in \GF[\tau,G]$, where $G$ is a fresh relation symbols 
of arity $w = \mathrm{max}\{ \mathrm{width}(\tau), \mathrm{width}(\varphi)\}$. 
First, we translate $\varphi$ to $\varphi'$ by replacing each 
clique-guarded quantifier occurring in $\varphi$ according to the pattern\footnote{where 
    in each individual case $G(\tup{x}\tup{y})$ is to be understood as referring to 
    the padding of $\tup{x}\tup{y}$ to a $w$-tuple, say, by repeated occurrences of the 
    last variable of the tuple $\tup{y}$. This is merely to render $G(\tup{x}\tup{y})$ 
    a well-formed atom; in the context of \eqref{eq_Gaxiom} the actual choice of padding 
    has no import.}
\begin{equation} \label{eq_CGFtoGFprimer}
\begin{array}{rcl}
  \left[\,(\exists \tup{y}.\alpha(\tup{x}\tup{y})) \,\psi\,\right]'  &=& 
  (\exists \tup{y}.G(\tup{x}\tup{y}))\,
\bigl( \alpha(\tup{x}\tup{y}) \land \psi' \bigr) \\
  \left[\,(\forall \tup{y}.\alpha(\tup{x}\tup{y})) \,\psi\,\right]'  &=& 
  (\forall \tup{y}.G(\tup{x}\tup{y}))\, 
\bigl( \alpha(\tup{x}\tup{y}) \limp \psi' \bigr)
\end{array}      
\end{equation}
otherwise trivially commuting with Boolean connectives.
The intended role of $G(\tup{z})$ is to reflect guardedness of $\tup{z}$ 
in the expanded signature $\tau \cup \{G\}$.
Accordingly, $\varphi^\ast$ is defined as the conjunction of $\varphi'$ and 
\begin{equation} \label{eq_Gaxiom}
  \bigwedge_{R} \ \bigwedge_{\{\tup{u}\}\subseteq\{\tup{z}\}} \ 
     (\forall \tup{z}.\,R(\tup{z})) \ G(\tup{u}) 
\end{equation}
where $R$ ranges over $\tau \cup \{G\}$ and $\tup{z}$ and $\tup{u}$ 
are of the appropriate arity such that all variables in $\tup{u}$ 
also occur in $\tup{z}$.
The following properties of this translation are readily verified.  
\begin{enumerate}
\item \label{fact_CGF2GF_1} 
      $|\varphi^\ast| = \calO(|\varphi|) + |\tau| w^{\calO(w)}$ where $w$ is as above. 
\item \label{fact_CGF2GF_2} Every model of $\varphi$ can be expanded to a model of $\varphi^\ast$  
      by interpreting $G$ as the universal relation of arity $w$. 
\item \label{fact_CGF2GF_3} Every \emph{conformal} model of $\varphi^\ast$ is also a model of $\varphi$:
      in conformal models every clique-guarded tuple is also guarded and   
      hence, by \eqref{eq_Gaxiom}, guarded by a $G$-atom. 
      All conformal models of \eqref{eq_Gaxiom} thus satisfy  
      $
        \forall \tup{z} \,\left(\,
             (G(\tup{z})\,\land\,\alpha(\tup{z})) \, \liff \, \alpha(\tup{z}) \,\right)
      $
      for every clique-guard $\alpha$ and, therefore, also $\varphi \liff \varphi'$. 
\end{enumerate}
These properties enable us to extend the scope of the reduction 
from mere satisfiability (in the finite) to the more general 
query entailment problem (in the finite).

\begin{lem} \label{lemma_CGF2GF_QA}
Let $\varphi \in \CGF$ and $q \in \UCQ$ be arbitrary, 
and let $\varphi^\ast \in \GF$ be the translation of $\varphi$ 
as explained above. Then  
$$
  \varphi \models_{(\mathrm{fin})} q \quad \iff \quad \varphi^\ast \models_{(\mathrm{fin})} q
$$ 
\end{lem}
\begin{proof}
Assume first that $\varphi^\ast \models_{(\mathrm{fin})} q$ 
and let $\frakA$ be a (finite) model of $\varphi$.
Then, by property~(\ref{fact_CGF2GF_2}) of the translation, 
$\frakA$ has an expansion $\frakA^\ast \models \varphi^\ast$. 
Obviously, $\frakA^\ast$ is finite whenever $\frakA$ is finite.
So by assumption, $\frakA^\ast \models q$, which trivially 
implies $\frakA \models q$, 
given that the interpretation of $G$ has no bearing on $q$. 
 
Assume now that $\varphi \models_{(\mathrm{fin})} q$ 
and take any (finite) model $\frakB \models \varphi^\ast$.
Let $w = \mathrm{width}(\frakB) \leq \mathrm{width}(\varphi^\ast)$, 
and let $\frakB^{(N)}$ be the $N$-th Rosati cover of $\frakB$ 
as in Theorem~\ref{thrm_main} for $N = w + 1$. 
Then $\frakB^{(N)} \models \varphi^\ast$. 
Furthermore, $\frakB^{(N)}$ is $N$-conformal and hence also conformal, since $N > w$.  
Therefore, by property (\ref{fact_CGF2GF_3}) of the translation, we have $\frakB^{(N)} \models \varphi$,
and so, by assumption, $\frakB^{(N)} \models q$, since $\frakB^{(N)}$ is finite 
whenever $\frakB$ is finite. 
Then, according to Fact~\ref{fact_covers_and_treeification}, $\frakB \models \chi_q$, 
and via Lemma~\ref{lemma_treeification} item~(ii) we conclude that $\frakB \models q$.  
\end{proof}

Combining the above with Theorem~\ref{thrm_FinControl} 
yields its generalisation to \CGF as follows. 

\newtheorem*{theorem_FinControl_CGF}{Theorem~\ref{thrm_FinControl_CGF}}
\begin{theorem_FinControl_CGF}
For every $\varphi \in \CGF$ and every $q \in \UCQ$ we have  
$\; \varphi \models q   \iff   \varphi \models_{\mathrm{fin}} q$. 
More specifically, if $\varphi \land \lnot\, q$ is satisfiable 
then it has a finite model of size $2^{(|\varphi|+|\tau|^{\calO(h)}) (wh)^{\calO(wh^2)}}$,
where $h$ is the height of $q$, $\tau$ is the signature of $\varphi$, 
and $w = \mathrm{max}\{ \mathrm{width}(\varphi), \mathrm{width}(\tau) \}$.
\end{theorem_FinControl_CGF} 
\begin{proof}
Theorem~\ref{thrm_FinControl} together with Lemma~\ref{lemma_CGF2GF_QA} 
provide the following chain of equivalences 
$$
 \varphi \models q \iff \varphi^\ast \models q \iff 
 \varphi^\ast \models_{\mathrm{fin}} q \iff \varphi \models_{\mathrm{fin}} q 
$$
proving the first assertion. 
Towards the size bound, if $\varphi \land \lnot q$ is satisfiable then, 
by Lemma~\ref{lemma_CGF2GF_QA}, so is $\varphi^\ast \land \lnot q$.
Recall that $|\varphi^\ast| = \calO(|\varphi|) + |\tau| w^{\calO(w)}$ 
by property~\eqref{fact_CGF2GF_1} of the translation.
According to Theorem~\ref{thrm_FinControl}, 
there exists a model $\frakB$ of $\varphi^\ast \land \lnot q$ 
of size 
$$
    2^{ (|\varphi^\ast|+|\tau|^{\calO(h)}) (wh)^{\calO(wh^2)} }
  = 2^{ (|\varphi|+|\tau| w^{\calO(w)}+|\tau|^{\calO(h)}) (wh)^{\calO(wh^2)} }
  = 2^{ (|\varphi|+|\tau|^{\calO(h)}) (wh)^{\calO(wh^2)} }
$$
Finally, as in the proof of Lemma~\ref{lemma_CGF2GF_QA} 
we construct the model $\frakB^{(N)}$ of $\varphi \land \lnot q$ 
by taking the $N$-th Rosati cover of $\frakB$ with $N=w+1$. 
By Theorem~\ref{thrm_main}, $|\frakB^{(N)}| = |\frakB|^{w^{\calO(w)}}$, 
which is still of the same order of magnitude 
$2^{ (|\varphi|+|\tau|^{\calO(h)}) (wh)^{\calO(wh^2)} }$ 
as $|\frakB|$, as claimed. 
\end{proof}

Theorem~\ref{thrm_small_models} as announced in the introduction is a 
straightforward corollary of the above.


\section{Complexity of query answering} \label{sec_Complexity}



In this paper \emph{query answering} is the problem of deciding $\varphi \models q$
for a given $\varphi \in \GFO$ and $q$ a \UCQ. 
By Lemma~\ref{lemma_treeification}~(ii) this amounts to testing unsatisfiability 
of the guarded sentence $\varphi \land \lnot \chi_q$, known to be \twoExpTime-complete 
and in {\sc DTime}($2^{\calO((r+|\varphi|+|\chi_q|)w^w)}$), 
where $r$ is the size and $w$ the width of $\tau$ \cite{Gr99JSL}.
With these parameters for $\tau$ recall from Lemma~\ref{lemma_treeification} 
that $|\chi^\tau_q| = r^{\calO(h)}(hw)^{\calO(hw)}$ and that $\chi^\tau_q$ is computable 
in time $|q|r^{\calO(h)}(hw)^{\calO(hw)}$ for any \UCQ $q$ of height $h$.  
Query answering is thus \twoExpTime-complete, even for a fixed query, and in 
{\sc DTime}($|q|r^{\calO(h)}(hw)^{\calO(hw)} + 2^{(r+|\varphi|)w^w+ r^{\calO(h)}(hw)^{\calO(hw)}}$).
Notice that there is a double-exponential dependence only in terms of the height $h$ of queries and 
the width $w$ of the signature.
Under increasing constraints on the variability of signatures we can break down 
and simplify the time complexity as follows:
\begin{itemize}
\item $2^{(|q||\varphi|)^{\calO(|q||\varphi|)}}$ 
      under no restrictions on $q$ nor on $\varphi$ nor on $\tau$; 
\item $|q|2^{(h|\varphi|)^{\calO(h|\varphi|)}}$ 
      without restrictions but highlighting the influence of query height;
\item $|q|(h|\varphi|)^{\calO(h)} + 2^{(h|\varphi|)^{\calO(h)}} 
       \leq |q| 2^{(h|\varphi|)^{\calO(h)}}$ 
      when the width of $\tau$ is bounded (a matching double-exponential 
      lower bound in this case follows from the work of Lutz \cite{Lutz07});
\item $|q| h^{O(h)} + 2^{\calO(|\varphi|) + h^{\calO(h)}} 
      \leq |q| 2^{\calO(|\varphi|) + h^{\calO(h)}}$ 
      for any fixed signature $\tau$ (for this case we provide a 
      single-exponential lower bound as stated in Theorem~\ref{thrm_complexity} 
      and proved in Proposition~\ref{prop_exptime} below);
\item $|q||\varphi|^{O(1)} + 2^{|\varphi|^{\calO(1)}} 
      \leq |q| 2^{|\varphi|^{\calO(1)}}$
      for queries of bounded height and over signatures of bounded width;
\end{itemize}

\noindent In \cite{Lutz07,Lutz08} Lutz considered the query answering problem against 
specifications in various description logics, among them a certain $\mathcal{ALCI}$, 
which can be naturally seen as a fragment of \GF. Lutz proved that answering \BCQ 
against $\mathcal{ALCI}$ specifications is \twoExpTime-complete. 
Given that description logics are interpreted over relational structures 
involving unary and binary predicates only, this implies that query answering 
against \GF is \twoExpTime-complete already for \BCQ and on signatures of width two. 

A further important particular case is that of \emph{acyclic queries}. 
Below \ACQ are unions of acylcic Boolean conjunctive queries. 
Observe that for $q$ an \ACQ the exponential blow-up in passing from $q$ 
to $\chi^\tau_q$ can be avoided by rewriting $q$ as a guarded existential 
sentence $q^\ast$ of essentially the same length as $q$. 
Query answering for \ACQ reduces in polynomial time to \GFO-satisfiability.
Regarding query answering against a fixed $\varphi \in \GFO$ we thus find 
that for \ACQ the complexity reduces to \ExpTime. 
In fact, it can also be shown to be \ExpTime-complete for certain $\varphi$,
cf. Proposition~\ref{prop_exptime} below. \\

For a fixed $\psi \in \GFO[\tau \cup \sigma]$ \emph{target query answering} 
is the problem of deciding $D \land \psi \models q$ on input $q$ a \UCQ 
and $D$ a $\tau$-structure (given as a conjunction of ground atoms with 
elements of $D$ as individual constants).
The next theorem summarises our observations on query answering and 
some results on subproblems of target query answering. \\

\begin{thm} ~ \label{thrm_complexity}
\begin{enumerate}
\setlength{\itemsep}{0.7ex}
\item Deciding $\varphi \models q$, on input $\varphi \in \GFO$ and $q$ a \UCQ, 
      is \twoExpTime-complete already for a fixed query $q$~\cite{Gr99JSL}, 
      or with the width of $\varphi$ bounded and $q$ a \BCQ~\cite{Lutz07}.
\item For each $\varphi \in \GFO$, deciding $\varphi \models q$ 
      on input $q$ an \ACQ is in {\sc ExpTime};  
      and it is \ExpTime-complete for certain $\varphi$.
\item There is a $\GF$-sentence $\psi$ such that deciding $D \land \psi \models Q$, 
      on input $Q$ a \BCQ and $D$ a conjunction of atoms \emph{of bounded width}, 
      is {\sc PSpace}-hard.
\item For all \emph{universal} $\psi \in \GFO$, deciding $D \land \psi \models q$, 
      on input $q$ a \UCQ and $D$ a conjunction of atoms, is in $\Pi^P_2$;  
      and for certain \emph{universal} $\psi$ it is $\Pi^P_2$-complete 
      already for \CQ $q$. 
\item For all $\psi \in \GFO$ and $q$ a \UCQ, deciding $D \land \psi \models q$ 
      on input $D$, is in {\em co-NP} and {\em co-NP}-complete already for $q=\bot$ 
      and certain \emph{universal} $\psi$. 
      Hence, satisfiability of $D \land \psi$ on input $D$ is in {\sc NP}, 
      and is {\sc NP}-complete for certain \emph{universal} $\psi \in \GFO$.
\end{enumerate}
\end{thm}

\noindent Observe that item (2) implies that satisfiability for $\GF$ can be 
\ExpTime-complete already for a fixed signature, see 
Proposition~\ref{prop_exptime} below. 
This strengthens a result of \cite{Gr99JSL}, where \ExpTime-completeness 
of satisfiability was shown for \GFO formulas over bounded arity but 
variable signatures.

\begin{cor} \label{coroll_exptime} 
For some relational signature $\tau$ satisfiability for $\GF[\tau]$ 
is \ExpTime-complete.
\end{cor}

Note that item (1) is but a restatement of results of Gr\"adel \cite{Gr99JSL} 
and Lutz \cite{Lutz07}, listed here for the sake of completeness. 
Similarly, the upper bound in item (2) is a consequence of \cite{Gr99JSL} 
as remarked at the beginning of this section. 
We prove each of the remaining claims separately in the propositions to follow.

\begin{prop} \label{prop_exptime}
There is a $\GFO$-sentence $\varphi$ such that deciding 
$\varphi \models Q$ for \ACQ{} $Q$ is \ExpTime-hard.
\end{prop}

\begin{proof}
It is well known that deterministic exponential time equals alternating 
polynomial space. We show how to encode the behaviour of polynomial-space 
alternating Turing machines into the query answering problem over a fixed 
formula. The formula $\varphi$ in question will merely provide the means 
of alternation blindly generating all trees potentially suitable for 
encoding any strategy of the existential player in any alternating run 
of any ATM on any input. 
Then, for any given ATM $M$ and input $w$ we craft an appropriate \UCQ $q_{M,w}$ 
comprising as disjuncts various conjunctive queries, each matching a 
different source of error in the encoding of the behaviour or acceptance 
of $M$ on input $w$. 
Ultimately the goal is to have $\varphi \models q_{M,w}$ if, and only if, 
$M$ has no accepting run on $w$, i.e., if the existential player has no 
winning strategy in the game corresponding to the computation of $M$ on input $w$. 

By standard arguments we may restrict attention to normalised ATM 
in which universal and existential states alternate in any run, 
which have a single universal initial state, disjoint sets of accepting 
and rejecting states, and every configuration of which has precisely two 
successor configurations (including accepting and rejecting configurations, 
which have only accepting or rejecting successor configurations, respectively).
Moreover we may assume that on each input of length $n$ the run of the 
normalised ATM uses precisely $p(n)$ amount of space for a polynomial $p$.

\noindent
Let $\varphi$ be the conjunction of the following guarded formulas, where $\oplus$ denotes exclusive or.
\begin{equation}\label{eq_phi_for_ATM}
\begin{array}{rl}
  (\ \exists x.\,R(x)\ )    &  B(x) \land A(x) \\
  (\,\forall x.\,B(x)\ )    &  E(x) \oplus A(x) \\
  (\,\forall x.\,B(x)\ )    &  \top(x) \oplus \bot(x)  \\
  (\,\forall x.\,B(x)\ )    &  T(x) \, \oplus \, \exists y. S(x,y)  \\
  (\ \forall xy.\,S(x,y)\ ) &  B(y) \, \land \, (\, E(x) \liff E(y)\,) \\
  (\ \forall x.\,T(x)\ )    &  E(x) \limp \exists y. F(x,y)  \\
  (\ \forall x.\,T(x)\ )    &  A(x) \limp \exists y_1. A_1(x,y_1)  \\ 
  (\ \forall x.\,T(x)\ )    &  A(x) \limp \exists y_2. A_2(x,y_2) \\
  (\ \forall xy.\,F(x,y)\ )   & B(y) \land A(y) \\
  (\ \forall xy.\,A_1(x,y)\ ) & B(y) \land E(y) \\
  (\ \forall xy.\,A_2(x,y)\ ) & B(y) \land E(y) \\
\end{array}
\end{equation}
Intuitively speaking, every model of $\varphi$ (or rather its guarded unravelling) 
represents a tree whose vertices correspond to instances of the variables $x$ and $y$ 
in the above formulation. There are three kinds of vertices: 
plain $B$-vertices, the root $R$-vertex and $T$-vertices. 
Each vertex represents a bit, hence the letter $B$, 
whose value is either $\top$ or $\bot$ as witnessed by the predicates of the same name.
There are four kinds of successor edges in every such tree corresponding to a 
model of $\varphi$: $S$,$F$,$A_1$ and $A_2$-successors. 
Every vertex is either a $T$-vertex or it has an $S$-successor.
The intention is that maximal $S$-successor chains connecting $T$-vertices 
encode individual configurations. 
In this sense a $T$-vertex is a terminal vertex of the configuration it belongs to 
and encodes its last bit. The predicates $E$ and $A$ mark whether the state 
of a given configuration is existential or universal, respectively. 
All vertices belonging to the same configuration (in between consecutive $T$-vertices) 
carry the same $E$ or $A$ marking, which is passed down along $S$-edges.
A $T$-vertex terminating a configuration in an existential state 
has an $F$-successor vertex beginning a new configuration.
A $T$-vertex belonging to a configuration with a universal state 
has both an $A_1$- and an $A_2$-successor vertex, 
each starting a successor configurations.

Note that for query answering one can always restrict attention to \emph{minimal models} 
of $\varphi$, which have only one $R$-vertex and at any vertex have at most one successor 
of whatever kind required and no other kind of successor vertices, and every vertex of which 
is reachable from the $R$-vertex via a sequence of overlapping $S$-, $F$-, $A_1$- and $A_2$-atoms. 
Minimal models of $\varphi$ are well suited to encode strategies of the existential player
in any game determined by a normalised ATM and an input word.

Using the framework provided by (minimal) models of $\varphi$ the power of unions 
of conjunctive queries suffices to filter out those models that do \emph{not} represent 
a winning strategy for the existential player in the game defined by a given ATM $M$ 
on a given input $w$. 
To demonstrate this we must first choose an appropriate encoding of Turing machine 
configurations. As is customary we write $\alpha q \beta$ for the configuration with 
tape contents $\alpha \beta$ when the machine is in state $q$ and its head is positioned 
on the first letter of $\beta$. 
Wlog.~the tape alphabet is binary, i.e.~$\alpha,\beta \in \{0,1\}^\ast$.

A configuration $\alpha q \beta$ will be encoded as a bit string 
$\tilde{\alpha} \tilde{q} \tilde{\beta}$, where $\tilde{\alpha}[2i] = \alpha[i]$ 
and $\tilde{\alpha}[2i-1] = 0$ for all $1 \leq i \leq |\alpha|$, and similarly 
for $\beta$ and $\tilde{\beta}$, and where $\tilde{q} = (11)^{j}(10)^{r-j}$ 
if $q$ is the $j$-th of the $r$ many states of $M$ in some fixed enumeration. 
In brief: the state is encoded in unary interleaved with $1$ digits at the point of the 
head position and around it the tape contents are interleaved with $0$ digits to 
clearly identify the position where the state is encoded. 

Let $M$ have $r$ states and use precisely $n=p(|w|)$-space on inputs of size $|w|$. 
The query $q_{M,w}$ will then consist of a disjunction of acyclic conjunctive queries 
(each of size $\calO(p(|w|))$ exhausting the reasons a model of $\varphi$ could fail 
to encode a winning strategy for the existential player in the game of $M$ on $w$:
\begin{itemize}
\item \CQ asserting the existence of an $S$-successor chain connecting $T$-vertices
      too short to represent a configuration, 
      or the existence of too long an $S$-successor chain:\\
      \begin{align*}
        \exists x_0,\ldots,x_{m} \ F(x_0,x_1) \,\land \bigwedge_{i<m} S(x_i,x_{i+1}) \ \land \ T(x_m)
      \end{align*} 
      for $m<2r+2n$ and similarly with $A_1(x_0,x_1)$ and $A_2(x_0,x_1)$ or $R(x_1)$ 
      in place of $F(x_0,x_1)$, and
      \begin{align*} 
         \exists x_1,\ldots,x_{2r+2n+1} \bigwedge_{i<2r+2n+1} S(x_i,x_{i+1})
      \end{align*}
\item \CQ asserting that on odd positions of a maximal $S$-successor chain the bit values 
      do not constitute a word in $0^\ast 1^r 0^\ast$:
      \begin{align*} 
        \exists x_1,\ldots,x_{2r+2n} \bigwedge_{i<2r+2n} S(x_i,x_{i+1}) \ \land 
            \, \top(x_{2i-1}) \, \land \, \bot(x_{2j-1}) \, \land \, \top(x_{2l-1})
      \end{align*}
      for $1\leq i<j<l\leq r+n$, along with
      \begin{align*} 
        \exists x_1,\ldots,x_{2r+1} \bigwedge_{i<2r+2n} S(x_i,x_{i+1}) \ \land 
            \, \top(x_{1}) \, \land \, \top(x_{2r+1})
      \end{align*}
      etc.
\item \CQ asserting that the configuration beginning with the root vertex is not
      the initial configuration for the given input word: one \CQ for each bit 
      in the encoding of the initial configuration looking for a wrong bit value, such as 
      \begin{align*}
        \exists x_1,\ldots,x_{m} \ R(x_1) \land \bigwedge_{i<m} S(x_{i},x_{i+1}) \land \ \bot(x_{m})
      \end{align*}
      for $m<2r$ odd or 
      \begin{align*}
        \exists x_1,\ldots,x_{m} \ R(x_1) \land \bigwedge_{i<m} S(x_{i},x_{i+1}) \land \ \top(x_{m})
      \end{align*}
      for $m>2r$ odd, and 
      (assuming for simplicity that the initial state is the $1$st one) for $2<m\leq2r$ even; 
      further
      \begin{align*}
        \exists x_1,\ldots,x_{2r+2k} \ R(x_1) \land \bigwedge_{i<2r+2k} S(x_i,x_{i+1}) \land \ \bot(x_{2r+2k})
      \end{align*} 
      for $w[k]=1$ and similarly for $w[k]=0$;
\item \CQ checking that consecutive $S$-successor chains do not represent successor configurations, e.g., 
      by asserting that in a given tape position not under the head of the ATM 
      the bit values in the two configurations are not identical: for instance as in 
      \begin{align*}
         \exists x_{-1},x_0,x_1,\ldots,x_{2r+2n},x_{2r+2n+1},x_{2r+2n+2} 
           \bigwedge_{-1\leq i<2r+2n, i\neq 2l} S(x_i,x_{i+1}) \ \land \quad \\ 
           \quad F(x_{2l},x_{2l+1}) \, \land \, \bot(x_{-1}) \, \land \, \bot(x_1) \, \land \, \bot(x_3) \, 
           \land \, \top(x_2) \, \land \, \bot(x_{2r+2n+2})
      \end{align*}
      and, similarly, with $A_1(x_{2l},x_{2l+1})$ or $A_2(x_{2l},x_{2l+1})$ in place of $F(x_{2l},x_{2l+1})$
      and with the bit values at corresponding positions $x_2$ and $x_{2r+2n+2}$ swapped, 
      and all this for each $l \leq r+n$;
\item \CQ asserting that the $A_i$-successor configuration of a universal configuration
      was not derived by the $i$-th of the two applicable transitions; 
\item \CQ asserting the existence of a configuration in reject state:
      \begin{align*}
         \exists x_1,\ldots,x_{2r} \bigwedge_{1\leq i<r} S(x_{2i-1},x_{2i}) \land S(x_{2i},x_{2i+1}) \ 
         \land \, \top(x_{2i-1}) \, \land \\
         \bigwedge_{i\leq q} \top(x_{2i}) \land \bigwedge_{q<i\leq r} \bot(x_{2i})
      \end{align*} 
      for each rejecting state $q$.
\end{itemize}
Much as those examples illustrated above, all of the flaws in the encoding 
of an accepting run can likewise be expressed using acyclic conjunctive queries, 
polynomially many in total, and each of size $\calO(|M|+p(|w|))$.
\end{proof}

\begin{prop} \label{prop_pspace}
There is a $\GFO$-sentence $\psi$ (using constants) such that deciding $D \land \psi \models Q$ 
on input $Q$ a \BCQ and $D$ a conjunction of atoms (of bounded width) is {\sc PSpace}-hard.
\end{prop}

\begin{proof}
The proof is by reduction from QBF. 
Wlog.~we may assume that the QBF instances are sentences in prenex normal form 
\begin{equation} \label{eq_qcnf}
  \exists X_0 \,\forall X_1 \,\exists X_2 \ldots \forall X_{2m-1} \,\exists X_{2m} \ \vartheta
\end{equation}
with $\vartheta$ a 3CNF-formula with free variables among $\{X_0,X_1,\ldots,X_{2m}\}$.
To represent ``valuation strategies'' for the existentially quantified variables 
we use a variant of the formula~\eqref{eq_phi_for_ATM}. 
On the one hand, the encoding is greatly simplified to single bit ``configurations'' 
representing Boolean values. On the other hand, Boolean values need to be encoded 
as elements of the domain of models to allow for a stronger form of pattern matching. 
To that end we rely on constants $0$ and $1$ and encode bit values using a binary 
predicate $V(x,b)$ where $b$ is either $0$ or $1$.
Let $\psi$ be the conjunction of the following (where, again, $\oplus$ stands for exclusive or).
\begin{equation} \label{eq_psi_for_QBF}
\begin{array}{rl}
                                &  R(r, 1, 0) \\
  (\ \forall xtf.\,R(x,t,f)\ )  &  B(x,t,f) \land A(x) \\
  (\,\forall xtf.\,B(x,t,f)\ )  &  E(x) \oplus A(x)  \\
  (\,\forall xtf.\,B(x,t,f)\ )  &  V(x,t) \oplus V(x,f)  \\
  (\ \forall xtf.\,B(x,t,f)\ )  &  E(x) \limp \exists y. S(x,y,t,f) \\
  (\ \forall xtf.\,B(x,t,f)\ )  &  A(x) \limp (\,\exists y_1. S(x,y_1,t,f)\,) V(y_1,t) \\  
  (\ \forall xtf.\,B(x,t,f)\ )  &  A(x) \limp (\,\exists y_2. S(x,y_2,t,f)\,) V(y_2,f) \\
  (\ \forall xytf.\,S(x,y,t,f)\ ) &  B(y,t,f) \land E(x) \oplus E(y) \\
\end{array}
\end{equation}
We may think of models of $\psi$ (more precisely their guarded unravelling) 
as representing infinite Boolean valuation trees with all possible assignments 
encoded at every odd level and encoding a choice of a Boolean value at each node 
at an even distance from the root. 
In minimal models of $\psi$ nodes at even levels (those marked with $A$) 
have precisely two successors representing the two Boolean values, 
whereas all nodes at odd levels (marked with $E$) have precisely one successor 
representing an existential choice of a Boolean value. 
For the encoding of Boolean valuations of the variables of~\eqref{eq_qcnf} 
only the first $2m$ levels of these trees will play a role. 

For each Q3CNF formula \eqref{eq_qcnf} with matrix $\vartheta$ consisting of $k$ clauses 
we assign an input data structure $D_\vartheta$ defined as the conjunction 
\[
  X(v_0,v_1) \land X(v_1,v_2) \land \ldots \land X(v_{2m-1},v_{2m}) 
\]
together with the conjunction of atoms $C(c_i, v_j, b)$ for all $1\leq i\leq k$ 
and $0 \leq j \leq 2m$ and $b\in\{0,1\}$ such that 
setting $X_j$ to $b$ does \emph{not} satisfy the $i$-th clause.
The elements $v_0,\ldots,v_{2m}$ and $c_1,\ldots,c_k$ of the ``database'' $D_\vartheta$ 
are perceived as constants, however, unlike $0$ and $1$, in $\psi$ they are not 
accessible by (constant) names. 

As in the proof of Proposition~\ref{prop_exptime} we design queries $Q_\vartheta$ 
to hold true in a model precisely when it contains a branch corresponding to 
an assignment of the variables falsifying (one of the clauses in) $\vartheta$. 
In fact, we use $Q_\vartheta$ merely as a general clause checking pattern, 
hence our reduction will use the same query $Q_\vartheta=Q_m$ for every Q3CNF 
formula with $2m$ alternately quantified variables:
\begin{align*}
   \exists \tup{x},\tup{y},\tup{z},t,f,c  \ R(y_0,t,f) \ \land \, 
     \bigwedge_{i<2m} \left(\, 
       X(x_i,x_{i+1}) \, \land \, S(y_i,y_{i+1},t,f) \, \land \, V(y_i,z_i) \, \land C(c,x_i,z_i) \, \right)
   \ .
\end{align*}
In other words, the query $Q_m$ matches the variables $\tup{x}$ to the corresponding constants $\tup{v}$
of $D$ and the variables $\tup{y}$ to a single branch in any given model of $D \land \psi$ 
and $\tup{z}$ to the sequence of Boolean values assigned to the variables $X_0, \ldots, X_{2m}$ 
on that branch. 
Then $Q_m$ is satisfied in a given model (which encodes a particular choice of Skolem functions 
for the existentially quantified variables $X_0,X_2,\ldots$) iff there is a clause $c$ that is 
falsified by an assignment to the universally quantified variables corresponding to some path 
within the model. Therefore, $Q_m$ is true in all models of $D_\vartheta \land \psi$ 
iff \eqref{eq_qcnf} is false.
\end{proof}

\begin{prop} \label{prop_piptwo}
The problem of deciding $D \land \psi \models Q$ for a fixed universal $\GFO$ sentence $\psi$ 
on input consisting of a conjunction of atoms $D$ and a \UCQ $Q$ is in $\Pi^P_2$ for each $\psi$, 
and for some $\psi$ is $\Pi^P_2$-complete already for \BCQ{} $Q$. 
\end{prop}

\begin{proof}
As in the previous propositions we may restrict attention to minimal models of $D \land \psi$
as regards query answering. Note that now, due to the universality of $\psi$, 
all minimal models of $D \land \psi$ have the same universe as $D$.
Hence, to check $D \land \psi \models Q$, one merely has to universally choose a model 
of $\psi$ on the universe of $D$ and existentially guess elements realising $Q$ 
in the model chosen.
This involves first universally choosing then existentially guessing a polynomial number 
of bits in terms of $|D|$, whence membership in $\Pi^P_2$. 

We show $\Pi^P_2$-hardness by reduction from the validity problem 
for quantified propositional formulas of the form 
\begin{align*}
  \forall X_1,\ldots,X_n \ \exists X_{n+1},\ldots,X_{n+m} \ \vartheta \, ,
\end{align*}
with $\vartheta$ a 3CNF formula with free variables $X_1,\ldots,X_{n+m}$.

In the input structure $D=D_n$ the universally quantified variables $X_1,\ldots,X_n$ 
are encoded as a successor chain
\begin{align*}
  S(v_1,v_2) \land S(v_2,v_3) \land \ldots \land S(v_{n-1},v_{n}) \ \land \ 
  X(v_1,1,0) \land \ldots \land X(v_n,1,0) \, ,
\end{align*}
along with the entire table of satisfying assignments 
of each of the three-literal clauses 
\begin{align*}
  \bigwedge_{(x,y,z)\in\{0,1\}^3}\bigwedge_{(i,j,k)\in\{0,1\}^3\setminus\{(x,y,z)\}} R_{(x,y,z)}(i,j,k) \ .
\end{align*}
The formula $\psi$ is devised so that (minimal) models of $D \land \psi$ 
will correspond to all possible assignments of the variables $X_1,\ldots,X_n$, for whatever $n$. 
We set 
\begin{align*}
  \psi \ = \ (\,\forall v.\,X(v,t,f)\,)\  V(v,t) \oplus V(v,f) \, .
\end{align*}
Finally, we use the query $Q_\vartheta$ to guess Boolean values 
for the existentially quantified variables $X_{n+1}, \ldots,X_{n+m}$ and 
to test satisfaction of $\vartheta$ by checking all triplets of Boolean values 
for variables occurring together in a clause 
against the truth table of the type of that clause.~Let 
\begin{align*}
   Q_\vartheta = 
   \exists z_1,\ldots,z_n, x_1,\ldots,x_{n+m} \ 
    \bigwedge_{1\leq i<n} \left( S(z_i,z_{i+1}) \land \, V(z_i,x_i) \right) \ \land \ V(z_{n},x_{n}) \ \ \land \\
    \bigwedge_{\tiny \begin{array}{c}(-1)^p X_i \lor (-1)^q X_j \lor (-1)^r X_k \\ 
                     \text{ a clause in } \vartheta\end{array}} 
               R_{(p,q,r)}(x_i,x_j,x_k) \ .
\end{align*}
It is now easy to verify that $D_n \land \psi \models Q_\vartheta$ \! iff \! 
$\forall X_1,\ldots,X_n \ \exists X_{n+1},\ldots,X_{n+m} \ \vartheta$ is valid. 
\end{proof}

\begin{prop} \label{prop_np}
For any fixed $\psi \in \GFO$ and input $D$ a conjunction of atoms, satisfiability 
of $D \land \psi$ is in NP and is NP-complete for certain universal $\psi$.
Deciding $D \land \psi \models Q$ on input $D$ for any fixed $\psi$ as before 
and fixed \UCQ $Q$ is in co-NP; and it is co-NP-complete already for $Q=\bot$ 
and certain universal $\psi$.
\end{prop}

\begin{proof}
To test satisfiability of $D \land \psi$ for a fixed $\psi$ one can pre-compute 
the Scott normal form $\Psi$ of $\psi$ as in \eqref{eq_Scott_nf}, 
as well as the set of admissible guarded atomic types in its signature, 
i.e.~those atomic types that can be realised in some model of $\Psi$. 
Having done that, for each given input $D$ it remains to be verified 
that its atoms can be assigned admissible atomic types wrt.~$\Psi$ 
such that (i) the type assigned to each atom actually contains that atom;  
(ii) overlapping atoms are assigned consistent types restricted to their overlap;
and (iii) the resulting structure satisfies $\Psi$. 
Such an assignment can be guessed and verified in NP.

To show NP-hardness of satisfiability of $D \land \psi$ for an appropriate $\psi$ 
observe that $3$-colourability of graphs can be directly formalised within this problem. 
Every simple graph $G=(V,E)$ can be identified with the conjunction $\bigwedge_{(u,v)\in E} E(u,v)$ 
with vertices $v\in V$ perceived as distinct constants. 
Then $G$ is $3$-colourable iff $G \land \psi$ is satisfiable, where  
\begin{equation}\label{eq_3CNF}
  \psi = \forall v  \bigvee_{i<3} \big(\, C_i(v) \land \bigwedge_{i\neq j <3} \lnot C_j(v) \,\big)
         \ \land \ 
         \left(\,\forall u v. E(u,v)\,\right) \ 
             \bigwedge_{i} \lnot \big(\, C_i(u) \land C_i(v) \, \big) \ .
\end{equation}

Turning to the problem $D \land \psi \models Q$ for fixed $\psi$ and $Q$ and input $D$, 
note that this is equivalent to the unsatisfiability of $D \land \psi \land \lnot \chi_Q$,
where $\chi_Q$ is the ``treeification'' of $Q$, that can now be precomputed since $Q$ is fixed. 
Thus, by the above argument, $D \land \psi \models Q$ on input $D$ can be verified in co-NP 
for any fixed $\psi$ and $Q$; and it is co-NP-complete for $\psi$ of \eqref{eq_3CNF} and $Q=\bot$. 
\end{proof}


\section{Canonisation and capturing} \label{sec_Capture}


The abstract version of the \emph{capturing Ptime problem} asks for 
an effective (recursive, syntactic) representation of all polynomial-time 
computable Boolean queries over finite relational structures. 
The core problem is not so much the effective representation of the class 
of all polynomial-time algorithms, which is easy via polynomially clocked  
Turing machines, say. Rather it lies in the requirement that these machines 
or algorithms must represent \emph{queries} on finite structures, i.e., 
they need to respect isomorphism in the sense of producing the same answer 
on isomorphic structures (or on inputs that represent isomorphic structures). 
In notation to be used below, we indicate this constraint explicitly in 
writing $\PTime/{\simeq}$ for the set of those \PTime algorithms that respect 
isomorphism. 

This undecidable semantic constraint is almost trivially enforcible and 
therefore mostly goes unnoticed when dealing with queries on \emph{linearly ordered}  
finite structures. This is because there is an obvious canonisation procedure 
on the class of all linearly ordered $\tau_<$-structures (we let $< \in \tau_<$ 
be the distinguished binary relation that is interpreted as a linear ordering 
in the class $\calC^<$ of all linearly ordered finite $\tau_<$-structures). 

By \emph{canonisation} (w.r.t.\ isomorphism over $\calC$) we here mean a map 
$\canon \colon \calC \rightarrow \calC$ such that $\canon(\frakA) \simeq \frakA$ 
and $\canon(\frakA) = \canon(\frakA')$ whenever $\frakA \simeq \frakA'$. 
Indeed, for $\frakA \in \calC^<$ we may just identify the linearly ordered 
universe of $\frakA$, $(A,<^\frakA)$, with an initial segment of $(\NN,<)$ 
to obtain such a canonical representative of the isomorphism type of $\frakA$.
Then the application of arbitrary polynomial time decision procedures to 
$\canon(\frakA)$ for $\frakA \in \calC^<$ -- i.e., the application of 
semantically unconstrained algorithms after pre-processing with $\canon$ -- 
provides an effective representation of the class of all polynomial time 
computable Boolean queries on $\calC$. 
It is well known from the fundamental results of Immerman~\cite{Imm86} and 
Vardi~\cite{Var82} that this abstract capturing result finds a concrete 
logical counterpart in the logics \LFP (least fixpoint logic) and 
\IFP (inductive fixpoint logic). 
The open question whether \PTime may also be captured, abstractly or by some 
suitable logic, over all not necessarily ordered finite structures has driven 
much of the development of descriptive complexity in finite model theory. 
Interesting variations of this question concern
\begin{enumerate}[label=(\alph*)]
\item restricted classes of finite structures other than $\calC^<$; and 
\item rougher equivalence relations than $\simeq$.
\end{enumerate}

We point to the work of Grohe and his survey~\cite{Grohesurvey} 
for successes with larger and larger natural classes of structures
in the sense of~(a), and to~\cite{Otto99TCS} for a very simple but 
interesting capturing result in the sense of~(b) concerning 
bisimulation-invariant \PTime ($\PTime/{\sim}$).
The class $\PTime/{\sim}$ consists of those \PTime Boolean queries 
that respect bisimulation equivalence; it can be regarded as the 
class of \PTime queries in the modal world. 
In that case, canonisation is obtained through passage to a definably 
ordered version of the bisimulation quotient of the given structure, 
$\Inv^< (\frakA) := (\Inv(\frakA),<)$ where $\Inv(\frakA) = \frakA/{\sim}$
such that $\Inv(\frakA) \sim \frakA$ is trivially satisfied.
For reasons indicated above, the distinction between (canonical) standard 
representations and definably linearly ordered versions of structures 
is often blurred, and in fact immaterial for our concerns. 
Hence we may avoid explicit passage to a standard representation of a 
linearly ordered structure like $\Inv^< (\frakA)$ and seemingly weaken 
the requirement that $\Inv^< (\frakA) = \Inv^< (\frakA')$ for $\frakA \sim \frakA'$ 
to $\Inv^< (\frakA) \simeq \Inv^< (\frakA')$.

A polynomial time canonisation procedure $\canon \colon \calC \to \calC$ 
w.r.t.\ some equivalence $\approx$ (on~$\calC$) will always yield 
an abstract capturing result for the class of all \PTime computable 
Boolean queries on finite structures from $\calC$ that respect $\approx$, 
which we denote by $\PTime/{\approx}$ (over~$\calC$). 
Since pre-processing with the canonisation procedure can be performed
in \PTime and enforces $\approx$-invariance, it can be coupled with 
any effective representation of otherwise unconstrained 
\PTime decision algorithms to capture $\PTime/{\approx}$:
\[
  \PTime/{\approx} \;\equiv\; \PTime \circ \canon,
\]
in a notation that suggests how canonisation acts as a filter to 
guarantee the required semantic invariance. 
If canonisation produces linearly ordered output structures, which we 
indicate notationally as in $\canon^<(\frakA) = (\canon(\frakA),<)$, 
then $\PTime/{\approx}$ is in fact captured by \LFP or \IFP over the 
canonisation results by the Immerman--Vardi Theorem:
\[
  \PTime/{\approx} \;\equiv\; \LFP \circ \canon^<
  \;\equiv\; 
  \IFP \circ \canon^<.
\]

In the case of $\PTime/{\sim}$, \cite{Otto99TCS} correspondingly translates 
the abstract capturing result into capturing by a suitable extension of 
the modal $\mu$-calculus, which exactly matches the expressive power of \LFP 
over the (internally interpretable) linearly ordered canonisations 
$\canon^<(\frakA) := \Inv^<(\frakA)$ indicated above. 

Here we primarily want to provide an abstract capturing result 
for $\gbisim$-invariant \PTime, $\PTime/{\gbisim}$, which
corresponds to \PTime in the guarded world. This is achieved through a
polynomial canonisation w.r.t.\ guarded bisimulation equivalence
which produces (definably) linearly ordered representatives of the 
complete guarded bisimulation types of given finite relational structures.
This canonisation may be of interest beyond our application to 
the capturing issue. The proposed canonisation produces linearly ordered
output structures $\canon^<(\frakA)$ that are uniformly 
interpretable over powers of the original structures $\frakA$ in
\IFP and \LFP in a $\gbisim$-invariant manner. 
An adaptation of the approach of~\cite{Otto99TCS} therefore also entails 
a concrete logical capturing result by means of some higher-dimensional 
guarded fixpoint logic, but we do not pursue this here. 

We fix some terminology, similar to the one discussed e.g.\ in~\cite{OttoLNL}, 
which makes sense for arbitrary equivalences $\approx$ between structures; 
we are going to use these notions solely with reference to guarded bisimulation 
equivalence.

\begin{defi} \label{invcandef} ~
\begin{enumerate}
\item A \emph{complete invariant $\Inv$} for $\approx$ on $\calC$ 
      (with values in some set $\calD$) is a map 
      \[\begin{array}{rcl}
        \Inv \colon \calC &\longrightarrow& \calD \\
        \frakA &\longmapsto & \Inv(\frakA) 
      \end{array}\]  
      such that $\Inv(\frakA) = \Inv(\frakA')$ 
      for all $\frakA \approx \frakA' \in \calC$.%
      \footnote{As above, we use notation $\Inv^<(\frakA) := (\Inv(\frakA),<)$ 
         or $\canon^<(\frakA) := (\canon(\frakA),<)$ 
         to indicate an invariant or canonisation that produces as output 
         some relational structure, which is (definably) linearly ordered. 
         In this context we relax e.g.\ the condition $\Inv(\frakA) = \Inv(\frakA')$ 
         to $\Inv^<(\frakA) \simeq \Inv^<(\frakA')$ without any loss.}
\item An \emph{inversion of the invariant $\Inv$} is then a map 
      $F \colon \calD \rightarrow \calC$ that acts as a right inverse to $\Inv$: 
      $\Inv \circ F = \mathrm{id}$, or, equivalently, 
      $F(\Inv(\frakA)) \approx \frakA$ for all $\frakA \in \calC$.
\item \emph{Canonisation} w.r.t.\ $\approx$ over $\calC$ is a map 
      \[\begin{array}{rcl}
        \canon \colon \calC &\longrightarrow& \calC \\
        \frakA &\longmapsto & \canon(\frakA) 
      \end{array}\]
      such that $\canon(\frakA) \approx \frakA$ 
      and $\canon(\frakA) = \canon(\frakA')$ 
      for all $\frakA \approx \frakA' \in \calC$.%
        \addtocounter{footnote}{-1}\footnotemark
\end{enumerate}
\end{defi}

\noindent Note that canonisations are complete invariants whose values are 
structures of the original kind while an invariant in general may 
produce values of a different format. 
Clearly an inversion of an invariant always yields a canonisation
of the form $\canon := F \circ \Inv$.
Note also that (1) says that $\approx$ is induced by equality of $\Inv$-images. 

\newtheorem*{theorem_canon}{Theorem (Canonisation) \ref{thrm_canon}}
\begin{theorem_canon}
Guarded bisimulation equivalence on finite relational structures 
admits \PTime canonisation. 
More specifically, definably linearly ordered versions of the 
guarded bisimulation game invariants $\Inv^<(\frakA) := (\Inv(\frakA),<)$ 
discussed above are \PTime computable complete invariants w.r.t.\ $\gbisim$ 
and admit \PTime inversions $F$ such that $\canon^< := F^< \circ \Inv^<$ 
produces linearly ordered representatives from the $\gbisim$-class 
of every finite relational structure $\frakA$. 
The values of both maps, $\Inv^<(\frakA)$ and $\canon^<(\frakA)$ 
are uniformly \IFP- and \LFP-interpretable as ordered quotients 
over $\frakA^w$ in a $\gbisim$-invariant manner, for fixed relational 
vocabulary $\tau$ of width $w$.
\end{theorem_canon}

\begin{cor}[Capturing]
\label{capcor}
Guarded-bisimulation-invariant \PTime can be captured (admits an effective, 
syntactic representation) in the form
\[\begin{array}{rcl}
  \PTime/{\gbisim} &\equiv&  
      \PTime \circ \canon, \mbox{ or } \\
  \PTime/{\gbisim} &\equiv&  
      \LFP \circ \canon^< \;\;\equiv\;\; \IFP \circ \canon^<. 
\end{array}\]
\end{cor}

We fix a finite relational vocabulary $\tau$ of width $w$. 
For a finite $\tau$-structure $\frakA$, we let $\GamDom(\frakA) \subseteq A^w$ 
be the set of all maximal guarded tuples of $\frakA$ (with repetitions 
of components where appropriate, to uniformly pad tuples to arity $w$). 
The guarded bisimulation game graph  $\Gam(\frakA)$ introduced in Section~\ref{sec_prelim}
has $\GamDom(\frakA)$ as its set of vertices, unary predicates for atomic types, 
and edge relations $(E_\rho)_{\rho \in \Sigma}$ for the set $\Sigma$ of all 
partial bijections $\rho \subseteq \{1,\ldots, w\} \times \{1,\ldots, w\}$, 
where 
\[
  (\bar{a},\bar{b}) \in E_\rho
  \quad \mbox{ iff } \quad 
  a_i = b_j \mbox{ for all $(i,j) \in \rho$.}
\]

The guarded bisimulation invariant $\Inv(\frakA) := \Gam(\frakA)/{\sim}$, 
as discussed in Section~\ref{sec_prelim}, is obtained as the quotient of 
this game graph w.r.t.\ \emph{modal} bisimulation equivalence.%
  \footnote{Note that the $\sim$-equivalence classes of vertices in 
    $\Gam(\frakA)$ are precisely the $\gbisim$-equivalence classes 
    of the maximal guarded tuples in $\GamDom(\frakA)$.}
Passage from $\frakA$ to this quotient $\Inv(\frakA)$ almost provides a complete 
invariant for $\gbisim$ in the sense of Definition~\ref{invcandef}, but not quite:
clearly $\frakA \gbisim \frakA'$ implies $\Gam(\frakA) \sim \Gam(\frakA')$ and 
hence $\Inv(\frakA) = \Gam(\frakA)/{\sim} \simeq \Gam(\frakA')/{\sim} = \Inv(\frakA')$, 
but not $\Inv(\frakA) = \Inv(\frakA')$ as required. 
As discussed above, this defect is overcome as soon as we provide a definable 
linear ordering of the universes $\GamDom(\frakA)/{\gbisim}$ of $\Inv(\frakA)$ and thus 
turn them into linearly ordered complete invariants $\Inv^<(\frakA) := (\Inv(\frakA),<)$.

Such a definable linear ordering can be obtained in an inductive refinement 
process, which produces a sequence of pre-orderings $\preceq^i$ on $\GamDom(\frakA)$. 
This refinement process starts from an arbitrary but fixed ordering of the 
finite set of atomic types of $w$-tuples, which is uniformly imposed as a 
global pre-order so that  
\[
  (\bar{a} \preceq^0 \bar{b}
  \mbox{ and }
  \bar{a} \preceq^0 \bar{b}) 
  \;\; \Leftrightarrow \;\;
  \frakA,\bar{a} \;\;\gbisimn{0}\;\; \frakA,\bar{b},
\]
i.e., the equivalence relation induced by $\preceq^0$ is atomic equivalence, 
which is $\gbisimn{0}$-equivalence, of maximal guarded tuples. In other words, 
$\prec^0$ uniformly defines a linear ordering of the quotient $\GamDom(\frakA)/{\gbisimn{0}}$ 
(or, equivalently, of $\Gam(\frakA)/{\sim^0}$ in terms of modal bisimulation).
The refinement proceeds in such a manner that each level $\preceq^i$ induces 
a linear ordering of the quotient 
$\GamDom(\frakA)/{\gbisimn{i}}$ (or $\Gam(\frakA)/{\sim^i}$).
Then the inductive fixpoint of this refinement sequence produces a pre-ordering $\preceq$ 
that linearly orders the quotient $\GamDom(\frakA)/{\gbisim}$ (or $\Gam(\frakA)/{\sim}$) 
and thus yields the desired $\Inv^<(\frakA)$. 
 

Besides $\preceq^0$ we fix an arbitrary ordering on the set $\Sigma$ of edge labels 
in the game graphs $\Gam(\frakA)$. For $\bar{a} \in \GamDom(\frakA)$ we define 
Boolean incidence functions $\iota_{\rho,\beta}$ for $\rho \in \Sigma$ and  
$\beta \in G(\frakA)/{\gbisimn{i}}$ according to
\[
  \iota_{\rho,\beta}(\bar{a}) := 1 \;\;
  \mbox{ iff }
  \;\;
  \bigl\{ \bar{b} \in G(\frakA) \colon (\bar{a},\bar{b}) \in \rho \bigr\} 
  \cap \beta 
  \not= \emptyset.
\]

Note that the $\iota_{\rho,\beta}$-value precisely describes the 
existence or non-existence of a move along a $\rho$-edge to a position 
in the $\gbisimn{i}$-class $\beta$. A simple analysis of one round 
in the guarded bisimulation game shows that 
\[
  \frakA,\bar{a} \;\gbisimn{i+1}\; \frakA,\bar{a}'
  \quad \mbox{  iff } \quad  
  \iota_{\rho,\beta}(\bar{a}) = \iota_{\rho,\beta}(\bar{a}')
  \mbox{ for all $\rho \in \Sigma$ and $\beta \in \GamDom(\frakA)/{\gbisimn{i}}$.}
\]

It follows that the pre-ordering $\preceq^{i+1}$ defined by a lexicographic 
ordering of tuples w.r.t.\ $\iota_{\rho,\beta}$-values is as desired. 

\begin{defi}
Let $\Inv^< \colon \frakA \mapsto (\Inv(\frakA),<) := (\Gam(\frakA),\prec)/{\sim}$ 
be the linearly ordered version of the quotient of the guarded bisimulation 
game graph $\Gam(\frakA)$ described above. We now refer to this linearly 
ordered structure as the \emph{ordered guarded-bisimulation invariant of $\frakA$.}
\end{defi}

The following is then immediate from the preceding discussion and the fact 
that the inductive refinement process outlined above is naturally captured 
as an inductive fixpoint (in the sense of inductive fixpoint logic \IFP), 
and hence, buy the Gurevich--Shelah Theorem also by a least fixpoint process 
(in the sense of least fixpoint logic \LFP).

\begin{lem}
$\Inv^<$ provides a complete invariant w.r.t.\ guarded bisimulation equivalence 
on the class of all finite $\tau$-structures. 
Moreover, $\Inv^<(\frakA)$ is \PTime computable from $\frakA$ and uniformly 
interpretable in a $\gbisim$-invariant manner as a quotient over $\frakA^w$, 
where $w$ is the width of $\tau$, in \IFP and \LFP.
\end{lem}

In order to prove Theorem~\ref{thrm_canon} and, as our main goal, the abstract 
capturing result of Corollary~\ref{capcor}, it therefore suffices to provide a 
\PTime computable (and hence also \IFP- and \LFP-interpretable) inversion for 
the complete invariant $\Inv^<$. This, in combination with $\Inv^<$, produces a 
\PTime computable (and automatically \IFP- and \LFP-interpretable) canonisation 
as follows: 
\[\xymatrix{
  *++{\frakA} \ar@<-3pt>@{|->}@/_2pc/[rrrr]^{\canon^<} \ar@{|->}[rr]^{\Inv^<} && 
  *++{\Inv^<(\frakA)} \ar@{|->}[rr]^{F^<} &&
  *++{\makebox(1,1)[l]{$\canon^<(\frakA) = F^<(\Inv^<(\frakA)) =(\canon(\frakA),<)$}}
}\rule{6cm}{0pt}
\]\vspace{4 pt}

\noindent yields a linearly ordered representative of the $\gbisim$-equivalence 
class of $\frakA$. 

In fact, we obtain $\canon^<(\frakA)$ as an ordered version of 
the Rosati cover $\frakR_2(\Inv^<(\frakA))$ obtained from an ordered version 
of $\Inv(\frakA)$.
Indeed, from Theorem~\ref{thrm_main_from_gbis_invariant} we already know that $\frakR_2(\frakI)$ 
is a suitable candidate for inverting an ordered guarded bisimulation invariant 
$(\frakI,<)$, because $\Inv^<(\frakR_2(\frakI)) = (\frakI,<)$. 
It remains to define a canonical linear ordering of the Rosati cover. 
This is a trivial exercise given the term structure of the elements 
of $\frakR_2$. First, using the linear order of a given invariant $(\frakI,<)$ 
and the standard ordering of natural numbers we define a linear ordering of 
constants $c^j_{e,i}$, say, in a lexicographic manner applied to the 
corresponding tuples $(e,i,j)$. We also fix a similarly defined ordering 
of all function symbols $f^j_{\rho,i}$ arising from $\frakI$ as introduced 
in Section~\ref{sec_RosatiDef}. 
It is then straighforward to extend these to a linear ordering of all 
terms (of depth $2$ in the case of $\frakR_2(\frakI)$) by stipulating that 
constants, i.e.~terms of height zero precede all terms of height one, 
which in turn precede all terms of height two in the ordering; and that 
terms of the same height are ordered first according to their root 
function symbols, then according to their sets of subterms inductively. 
It is apparent that such an ordering of $\frakR_2(\frakI)$ can be computed 
in polynomial time given the ordered invariant $(\frakI,<)$.
Hence $\canon^<(\frakA) = (\frakR_2(\Inv^<(\frakA)),<)$ is well defined, 
polynomial-time computable and fulfills the claims of Theorem~\ref{thrm_canon}.

%
%


\bibliographystyle{latex8}


%
%

\end{document}